\documentclass[a4paper,12pt]{article}

\usepackage{complexity} 
\usepackage{enumerate}
\usepackage{amsmath, amssymb, amsthm,complexity}
\usepackage{authblk}
\usepackage{todonotes}

\newcommand{\sv}[1]{}
\newcommand{\lv}[1]{#1}

\usepackage{vmargin}
\setmarginsrb{1in}{1in}{1in}{1in}{0mm}{0mm}{0mm}{7mm}

\usepackage{mathtools}
\usepackage{thmtools}
\usepackage{thm-restate}
\usepackage{hyperref}
\usepackage{cleveref}
\usepackage{complexity}
\usepackage{authblk}
\theoremstyle{plain}

\usepackage{todonotes}
\usepackage{microtype}
\usepackage{multirow}
\usepackage{amsmath,amssymb}
\usepackage{comment}
\usepackage{enumerate}
\usepackage{xspace}
\usepackage{listings}
\usepackage{color}
\usepackage{cite}
\usepackage{graphicx}
\RequirePackage{fancyhdr}
 \usepackage{xcolor}
\usepackage{boxedminipage}
\usepackage[ruled]{algorithm2e}

\newcommand{\mypara}[1]{\smallskip\noindent{\textbf{\sffamily #1} \ }}

\newcommand{\col}{\mathrm{\mathop{col}}}

\newcommand{\nats}{{\mathbb{N}}}
\newcommand{\cO}{{\cal O}}

\newcommand{\cF}{{\mathcal{F}}}

\newcommand{\cR}{\mathcal{R}}
\newcommand{\cI}{\mathcal{I}}
\newcommand{\cA}{\mathcal{A}}

\newcommand{\cH}{\mathcal{H}}

\newcommand{\OPT}{\mbox{\rm OPT}}

\newcommand{\eps}{\varepsilon}

\newcommand{\shortversion}[1]{}

\newcommand{\pcvc}{{$p$-Edge-CVC}}

\newcommand{\CVC}{{\sc Connected Vertex Cover}}
\newcommand{\BCVC}{{\sc Biconnected Vertex Cover}}
\newcommand{\PECVC}{{\sc $p$-Edge-Connected Vertex Cover}}

\newcommand{\PVCVC}{{\sc $p$-Connected Vertex Cover}}

\newcommand{\bigoh}{{\mathcal O}}

\newcommand{\nka}{{NP $\subseteq$ coNP/poly}}
\newcommand{\nn}{{\mathbb N}}

\newcommand{\bbF}{{\mathbb F}}

\newenvironment{tightcenter}
 {\parskip=0pt\par\nopagebreak\centering}
 {\par\noindent\ignorespacesafterend}

\usepackage{tikz}
\usetikzlibrary{calc}
\usepackage{xargs}
\usepackage{xifthen}
\usepackage{framed}

\usepackage{ctable}

\newlength{\RoundedBoxWidth}
\newsavebox{\GrayRoundedBox}
\newenvironment{GrayBox}[1]%
   {\setlength{\RoundedBoxWidth}{\textwidth-4.5ex}
    \def\boxheading{#1}
    \begin{lrbox}{\GrayRoundedBox}
       \begin{minipage}{\RoundedBoxWidth}%
   }{%
       \end{minipage}
    \end{lrbox}%
    \begin{tightcenter}%
    \begin{tikzpicture}%
       \node(Text)[draw=black!60,fill=white,rounded corners,%
             inner sep=2ex,text width=\RoundedBoxWidth]%
             {\usebox{\GrayRoundedBox}};
        \coordinate(x) at (current bounding box.north west);
        \node [draw=white,rectangle,inner sep=3pt,anchor=north west,fill=white] 
        at ($(x)+(6pt,.75em)$) {\boxheading};
    \end{tikzpicture}
    \end{tightcenter}\vspace{0pt}%
    \ignorespacesafterend
}    

\newenvironment{problem-box}[2][]{\noindent\ignorespaces%
                                \FrameSep=6pt%
                                \parindent=0pt%
                \vspace*{-.5em}
                \ifthenelse{\isempty{#1}}{%
                  \begin{GrayBox}{\textsc{#2}}%
                }{%
                }
                \newcommand\Prob{{\sf Problem:}}%
                \newcommand\Input{{\sf Input:}}%
                \newcommand\Parameter{{\sf Parameter:} }          
                \begin{tabular*}{\textwidth}{@{\hspace{.1em}} >{\itshape} p{1.2cm} p{0.85\textwidth} @{}}%
            }{
                \end{tabular*}%
                \end{GrayBox}%
                \vspace*{-.5em}
                \ignorespacesafterend
            }

\newtheorem{theorem}{\bf Theorem}
\newtheorem{definition}{\bf Definition}
\newtheorem{proposition}{\bf Proposition}

\newtheorem{reduction rule}{\bf Reduction Rule}

\newtheorem{lemma}{\bf Lemma}

\def\DEBUG{true}

\ifdefined\DEBUG{}
\def\rem#1{{\marginpar{\raggedright\scriptsize #1}}}

\newcommand{\magnus}[1]{{{#1}}}
\newcommand{\magnusr}[1]{\rem{\textcolor{blue}{\(\bullet \) #1}}}

\newcommand{\bmp}[1]{{{#1}}}
\newcommand{\bmpr}[1]{\rem{\textcolor{purple}{\(\bullet \) #1}}}

\newcommand{\gutinr}[1]{\rem{\textcolor{orange}{\(\bullet \) #1}}}

\newcommand{\diptar}[1]{\rem{\textcolor{red}{\(\bullet \) #1}}}

\else
\newcommand{\magnus}[1]{#1}
\newcommand{\bmp}[1]{#1}

\newcommand{\magnusr}[1]{}
\newcommand{\bmpr}[1]{}
\newcommand{\gutinr}[1]{}
\newcommand{\diptar}[1]{}
\fi

\DeclareMathOperator{\rank}{rank}

\title{$p$-Edge/Vertex-Connected Vertex Cover: Parameterized and Approximation Algorithms}

\date{}
\author[1]{Carl Einarson}
\author[1]{Gregory Gutin}
\author[2]{Bart M. P. Jansen}
\author[3]{Diptapriyo Majumdar}
\author[1]{Magnus Wahlstr{\"o}m}
\affil[1]{Royal Holloway, University of London, Egham, United Kingdom\\
    \texttt{einarsoncarl@gmail.com,\{g.gutin|diptapriyo.majumdar|magnus.wahlstrom\}@rhul.ac.uk}}
\affil[2]{Eindhoven University of Technology, The Netherlands\\
  \texttt{b.m.p.jansen@tue.nl}}
\affil[3]{Indraprastha Institute of Information Technology Delhi, New Delhi, India\\
	\texttt{diptapriyo@iiitd.ac.in}}

\begin{document}

\maketitle

\begin{abstract}
We introduce and study two natural generalizations of the Connected Vertex Cover (VC) problem: the $p$-Edge-Connected and $p$-Vertex-Connected VC problem (where $p\ge 2$ is a fixed integer).
We obtain an $2^{\cO(pk)}n^{\cO(1)}$-time algorithm for $p$-Edge-Connected VC and an $2^{\cO(k^2)}n^{\cO(1)}$-time algorithm for $p$-Vertex-Connected VC.
Thus, like Connected VC, both \bmp{constrained} VC problems are FPT.  Furthermore, like Connected VC, neither problem admits a polynomial kernel unless {NP $\subseteq$ coNP/poly}, which is highly unlikely.
We prove however that both problems admit time efficient polynomial sized approximate kernelization schemes.
Finally, we describe a $2(p+1)$-approximation algorithm for the $p$-Edge-Connected VC.
The proofs for the new VC problems require more sophisticated arguments than for Connected VC. In particular, for the approximation algorithm we use Gomory-Hu trees and for the approximate kernels a result  on small-size spanning $p$-vertex/edge-connected subgraphs of a $p$-vertex/edge-connected graph by Nishizeki and Poljak (1994) and Nagamochi and Ibaraki (1992). 
\end{abstract}


%
\section{Introduction}
\label{sec:intro}
For a graph $G=(V,E)$, a set $C\subseteq V$ is a {\em vertex cover}
if for every edge $uv\in E$ at least one of the vertices $u,v$
belongs to $C$. The well-known classical {\sc Vertex Cover} problem
is the problem of deciding whether a graph $G$ has a vertex cover of
size at most $k.$ This problem is NP-complete and thus was studied from 
parameterized (it is usually parameterized by $k$) and approximation algorithms view points. 
 {\sc Vertex Cover} and its generalizations have been important in developing basic and advanced methods and approaches 
 for parameterized and approximation algorithms. Thus, it 
 was named the {\em Drosophila} of fixed-parameter algorithmics \cite{Niedermeier06}.
 \magnus{In particular, it is well known that {\sc Vertex Cover} is fixed-parameter tractable, admits a kernel with $2k$ vertices, and has a trivial 2-approximation~\cite{ChenKJ01,NemhauserT75}.}

 \magnus{A well-studied variant of \textsc{Vertex Cover} is \CVC{} (CVC), where the vertex cover is additionally required to be connected.
 This problem is FPT parameterized by the solution size $k$, with the fastest known algorithm running in time $\cO^*(2^k)$ due to Cygan~\cite{Cygan12} after a sequence of improvements\footnote{The $\cO^*$ notation suppresses the polynomial factors.}. 
 CVC also has a classic 2-approximation, implicit in Savage~\cite{Savage82}.
 However, unlike \textsc{Vertex Cover}, CVC was shown to have no polynomial kernel unless {\nka}~\cite{DomLS14}}.
\bmp{Through reductions, this lower bound can be shown to imply that several other variations of \textsc{Vertex Cover}, like the ones we will consider, do not admit polynomial kernels.} 

\magnus{To get around limitations on kernelization, Lokshtanov et al.~\cite{LokshtanovPRS17}, in pioneering work,
 proposed the notion of \emph{approximate kernels}. Although the details are somewhat technical, in essence
 an $\alpha$-approximate kernel can be thought of as a kernel that only preserves solution optimality up to a factor of $\alpha$,
 e.g., from an optimal solution to the output instance we can recover an $\alpha$-approximate solution to the input.
 Lokshtanov et al.~\cite{LokshtanovPRS17} considered several problems that are known not to admit polynomial kernels (unless  {\nka})
 and analyzed the existence of $\alpha$-approximate polynomial kernels.}
In particular, they proved that {\CVC} admits an $\alpha$-approximate polynomial kernel for every fixed $\alpha$.

\magnus{We study two natural variations of CVC under tighter connectivity constraints,
 from the perspectives of FPT algorithms, approximate kernels, and approximation algorithms.
 Let us review the definitions.}
 
A connected graph is called {\em $p$-edge-connected} if it remains connected whenever fewer than $p$ edges are deleted.
A connected graph is called {\em $p$-connected} if it has more than $p$ vertices and it remains connected whenever fewer than $p$ vertices are deleted.
In this paper, we introduce the following two natural generalizations of {\CVC}, the {\PECVC} and {\PVCVC} problems. For both problems, $p$ is a fixed positive integer. 

\begin{problem-box}[]{\PECVC (\pcvc)}
	\Input & An undirected graph $G$, and an integer $k$ \\
	\Parameter & \hspace{2 mm} $k$ \\
	\Prob & Does $G$ have a set of at most $k$ vertices that is a vertex cover and induces a $p$-edge-connected subgraph?
\end{problem-box}

\begin{problem-box}[]{\PVCVC ($p$-CVC)}
	\Input & An undirected graph $G$, and an integer $k$ \\
	\Parameter & \hspace{2 mm} $k$ \\
	\Prob & Does $G$ have a set of at most $k$ vertices that is a vertex cover and induces a $p$-connected subgraph?
\end{problem-box}

To the best of our knowledge, none of these problems were studied in the fields of parameterized and approximation algorithms before. These problems may be of interest in the following security setting.
We have a network whose links (edges) have to be protected by monitors positioned in its nodes against an attacker targeting the links. 
To allow us to overview all the links and inform all other monitors about the attack, the monitors should form a connected vertex cover $C$. 
Higher vertex and edge connectivities of $C$ provide a higher degree of resiliency in the monitoring system when monitors/links can be disabled.

Graph theoretic problems with $2$-vertex/edge-connectivity constraints have been studied before.
Li et al.~\cite{LiYW16} proved structural results for {\sc $2$-Edge-Connected Dominating Set} and {\sc $2$-Node Connected Dominating Set}.
Recently, Nutov \cite{Nutov} 
obtained an approximation algorithm with expected ratio $\bigoh(\log^4 n \cdot \log \log n \cdot (\log \log \log n)^3)$ for {\sc Biconnected Domination}.
Note that no study has investigated $p$-vertex/edge-connectivity constraints for $p > 2$.

\mypara{Our results.} Both {\PECVC} and {\PVCVC} are {\sf NP}-complete; proofs for this follow from the reductions provided in Theorems \ref{thm:pcvc-lowerbound} and \ref{thm:pecvc-lowerbound}, respectively.
It is also not hard to obtain simple FPT algorithms for both {\PECVC} and {\PVCVC}.
Unfortunately, the running time of such algorithms is $\cO^*(2^{\cO(k^2)})$. For the sake of completeness, we give a short proof for the existence of such an algorithm for {\PVCVC}.
However, for  {\PECVC} we can do better: our first main result is a single-exponential fixed-parameter algorithm for {\PECVC} using dynamic programming on matroids and some specific characteristics of $p$-edge-connected subgraphs.

\begin{restatable}{theorem}
{pECVCSinglyExponentialAlgo}
\label{thm:improved-algo-pecvc}
For every fixed $p \geq 2,$  {\PECVC} can be solved in $2^{\cO(pk)}n^{\cO(1)}$
  deterministic time and space.
\end{restatable}

Our algorithm for {\sc \pcvc} is as follows.
First, we enumerate all minimal vertex covers of $G$ of size at most $k$.
The number of such vertex covers is at most $2^k$, and they can be enumerated in $\cO^*(2^k)$ time, and space (see~\cite{MolleRR08}).
Then, for every minimal vertex cover $H$ of $G$, we use representative
sets to check if it can be extended to a $p$-edge-connected vertex cover $S^{\star} \supseteq H$ of size at most $k$.
This step uses a characterization of $p$-edge-connected graphs due to Agrawal et al.~\cite{AgrawalMPS17}.

Unfortunately, our approach to prove Theorem~\ref{thm:improved-algo-pecvc} does not work for {\PVCVC}.

{After showing that both {\pcvc} and $p$-CVC{} do not admit polynomial kernels unless {\nka}, we prove that there are $(1+\eps)$-approximate polynomial kernels for both problems
for every $\eps>0$. The kernelization algorithm is the same in both cases, but the analysis differs. We provide necessary terminology and notation on approximate kernels in Section \ref{sec:approxkerneldef}.
Our results on $(1+\eps)$-approximate kernels are as follows.}

\begin{restatable}{theorem}{pKernelMainTheorem}
\label{thm:psaks}
For every $\eps>0$ and every fixed $p \geq 2$, the problems {\PECVC} and {\PVCVC} admit a $(1+\eps)$-approximate kernel with $k + 2k^2 + \lfloor (3+\eps)k \rfloor k^{2\lceil p/ \min(\eps,1) \rceil}$ vertices.
\end{restatable}

\bmp{Note that the size guarantee is mainly of interest for~$\epsilon \leq 1$, due to the term~$\min(\eps,1)$.}
{The main difficulty in obtaining a kernelization for connected variants of \textsc{Vertex Cover} consists of dealing with vertices which are not needed to make a minimal vertex cover, but that may be needed to boost the connectivity of a minimal vertex cover. We give a marking procedure which selects a bounded number of vertices to preserve in the kernelized instance, and show that the role that any of the forgotten vertices plays to boost the connectivity of a vertex cover can be mimicked by a small set of marked vertices. \bmp{We use existing results by Nishizeki and Poljak~\cite{NishizekiP94} and Nagamochi and Ibaraki \cite{NagamochiI92} on sparse spanning $p$-vertex/edge-connected subgraphs of $p$-vertex/edge-connected graphs for this argument.} This leads to a proof that the size of an optimal solution does not increase by more than a factor~$(1+\eps)$ in the kernelized instance. From the redundance of the forgotten vertices for the purpose of making a small vertex cover it easily follows that any solution in the reduced graph is also a valid solution in the original, which leads to our lossy kernelizations.}


%


Our last main result is the following: 

\begin{restatable}{theorem}
{ApproxAlgoPEdgeCVC}
\label{thm:approx-algo-pecvc}
For every fixed $p \geq 2,$ {\PECVC} admits a \bmp{polynomial-time} $2(p+1)$-factor approximation algorithm.
\end{restatable}

The proof of this theorem uses the notion of $p$-blocks which can be obtained from a Gomory-Hu tree. 
Unfortunately, this approach is not applicable to {\PVCVC} and the existence of a constant-factor approximation algorithm for {\PVCVC} is an open problem (recall that $p$ is a fixed positive integer). 

\paragraph*{Related Work} Lokshtanov et al. \cite{LokshtanovPRS17} obtained a polynomial size approximate kernelization scheme (PSAKS\bmp{, see Definition~\ref{defn:psaks}}) for {\sc Connected Vertex Cover} parameterized by the solution size.
Majumdar et al. \cite{MajumdarRS20} considered parameters that are strictly smaller than the size of the solution and obtained a PSAKS for {\sc Connected Vertex Cover} parameterized by the deletion distance of the input graph to each of the following classes of graphs: cographs, bounded treewidth graphs, \bmp{and} chordal graphs. Ramanujan~\cite{Ramanujan19} obtained a PSAKS for {\sc Connected Feedback Vertex Set}. 
Eiben et al.~\cite{EibenHR19}  obtained a similar result for the {\sc Connected $\cH$-Hitting Set} problem, where $\cH$ is a fixed set of graphs. 
Eiben et al.~\cite{EibenKMPS19} also obtained PSAKSs for {\sc Connected Dominating Set} on two classes of sparse graphs. Krithika et al. \cite{KrithikaMR18} designed PSAKSs for the {\sc Tree Contraction}, {\sc Star Contraction}, {\sc Out-Tree Contraction}, and {\sc Cactus Contraction} problems. Dvor{\'{a}}k et al. \cite{DvorakFKMTV21} designed a PSAKS for {\sc Steiner Trees} parameterized by the number of non-terminals. 

Improving a result of Lokshtanov et al. \cite{LokshtanovPRS17}, Manurangsi \cite{manurangsi2018} obtained a smaller size PSAKS for {\sc Max $k$-Vertex Cover}, where given an edge-weighted graph $G$ and an integer $k,$ and the aim is to find a subset $S$ of $k$ vertices that maximizes the total weight of edges covered by $S$ (an edge $e$ is covered by $S$ if at least one \bmp{endpoint} of $e$ is in $S$). Lossy kernels were also obtained for problems of contraction to generalizations of trees by Agarwal et al. \cite{AgarwalST19} and to classes of chordal graphs by Gunda et al. \cite{GundaJLST21}. 
A lossy kernel was designed by Bandyapadhyay et al. \cite{BandyapadhyayFGP21}
for {\sc Same-Size Clustering} parameterized by the cost of clustering.
Recently, Jansen and Wlodarzyck \cite{JansenW22} obtained a $510$-factor approximate kernel of polynomial size for {\sc Planar Vertex Deletion}. 




\paragraph*{Organization} The rest of the paper is organized as follows.
In Section~\ref{sec:prelim}, \bmp{we provide additional terminology and preliminaries needed for our algorithms and hardness proofs. Section~\ref{sec:improved-pedge-cvc} presents the FPT algorithms for the vertex- and edge-connectivity versions of the problem, proving Theorem~\ref{thm:improved-algo-pecvc}. Section~\ref{sec:psaks} presents the approximate kernelization schemes, leading to a proof of Theorem~\ref{thm:psaks}. The approximation algorithm of Theorem~\ref{thm:approx-algo-pecvc} is presented in Section~\ref{sec:approx-p-edge-cvc}. Section~\ref{sec:nopolykernel} contains the kernelization lower bounds which justify the use of approximate kernels.}
We conclude the paper with Section \ref{sec:conc}, where we discuss some open problems on the topic.

\section{Preliminaries}
\label{sec:prelim}
\subsection{Sets and Graph Theory}
For $r \in \nn$, we use $[r]$ to denote the set $\{1,2,\ldots,r\}$. 
Let $U$ be a set of elements, and $\cF$ a family of subsets of $U.$ (A {\em family} may have multiple copies of the same subset.) 
Then $\cF$ is said to be a {\em laminar family}, if for every $X, Y \in \cF$, either $X \subseteq Y$, or $Y \subseteq X$, or $X \cap Y = \emptyset$.


\bmp{For a graph~$G$ and subset $S \subseteq V(G)$}, $G[S]$ denotes the subgraph of $G$ induced by $S$.
Given a connected graph $G = (V, E)$, a set of vertices $S \subseteq V(G)$ is called a {\em separator} of $G$ if $G - S$ is a disconnected graph.
For two disjoint vertex sets $A$ and $B$ of $G,$ an $(A,B)$-{\em cut} is a set $F$ of edges such that $G-F$ is disconnected and no connected component of $G-F$ contains both a vertex of $A$ and a vertex of $B$.
{Given a connected graph $G = (V, E)$, for a proper subset $A$ of $V(G)$, let $\bar A = V(G) \setminus A$. Then, the {\em $(A, \bar{A})$-cut} is the set of edges with one endpoint in $A$ and the other endpoint in $\bar A$. 
The {\em size} of an $(A, \bar A)$-cut is the number of edges in this set, denoted by $|(A, \bar A)|$.
A {\em cut} is {a set of edges equal to} the $(A, \bar A)$-cut for some proper subset $A$ of $V(G)$.}

A graph $G$ is {\em $p$-connected} if it has more than~$p$ vertices and $G-X$ is connected for every $X\subseteq V(G)$ with $|X|<p.$
A graph $G$ is {\em $p$-edge-connected} if it has at least two vertices and $G-Y$ is connected for every $Y\subseteq E(G)$ with $|Y|<p.$
By Menger's theorem, (i) a graph is  $p$-connected \bmp{if and only if} 
there are at least $p$ internally vertex-disjoint paths between every pair of vertices, and (ii)  a graph is  $p$-edge-connected  \bmp{if and only if} 
there are at least $p$ edge-disjoint paths between every pair of vertices.

\bmp{The following proposition will be useful.}

\begin{proposition}~\cite{NishizekiP94,NagamochiI92}
\label{thm:connectivity-certificate}
Let $G$ be a $p$-vertex/edge-connected graph \bmp{for some~$p\geq 1$.}
Then, there exists a polynomial-time algorithm that computes a $p$-vertex/edge-connected spanning subgraph $H$ of $G$ such that $H$ has at most $p|V(G)|$ edges.
\end{proposition}



\subsection{Parameterized Algorithms and Kernels}
A {\em parameterized problem} $\Pi$ is a subset of $\Sigma^* \times \nats$ for some finite alphabet $\Sigma$.
An instance of a parameterized problem is a pair $(x,k)$ where $k$ is called the {\em parameter} and $x$ is the {\em input}.
\begin{definition}[Fixed-Parameter Tractability]
\label{defn:fpt}
A parameterized problem $\Pi \subseteq \Sigma^* \times \nats$ is said to be {\em fixed-parameter tractable} (or {\em FPT}) if there exists an algorithm for solving the problem $\Pi$ that on input $(x,k)$, runs in $f(k)|x|^c$ time where $f: \nats \rightarrow \nats$ is a computable function and $c$ is a constant.
\end{definition}


\lv{\begin{definition}[Kernelization]
\label{defn:kernel}
\bmp{Let~$\Pi \subseteq \Sigma^* \times \mathbb{N}$ be a parameterized problem. A \emph{kernelization algorithm}, or in short, a \emph{kernelization}, for~$\Pi$ is an
algorithm with the following property. For any given~$(x, k) \in \Sigma^* \times \mathbb{N}$, it outputs in time polynomial in~$|x| + k$ a string~$x' \in \Sigma^*$ and an integer~$k' \in \mathbb{N}$ such that:
$$((x, k) \in \Pi \Leftrightarrow (x', k') \in \Pi) \text{ and } |x'|, k' \leq h(k)$$
where h is an arbitrary computable function. If~$\cA$ is a kernelization for~$\Pi$,
then for every instance~$(x, k)$ of~$\Pi$, the result of running~$\cA$ on the input~$(x, k)$ is called the kernel of ~$(x, k)$ (under~$\cA$). The function~$h$ is referred to
as the size of the kernel. If~$h$ is a polynomial function, then we say that the kernel is polynomial.}
%
\end{definition}}
It is well-known~\cite{CyganFKLMPPS15} that a decidable parameterized problem is FPT if and only if it has a kernel. Kernelization can be viewed as a theoretical foundation of computational preprocessing and thus we are interested in investigating when a parameterized problem admits a kernel of small size. 
We are especially interested in {\em polynomial kernels} for which the bound $g(k)$ is a polynomial and thus in classifying which parameterized problems admit polynomial kernels or not. 
Over the last decade, the area of kernelization has developed a large number of approaches to design polynomial kernels as well as a number of tools for proving lower bounds based on assumptions from complexity theory. 
The lower bounds rule out the existence of polynomial kernels for many parameterized problems. We refer the reader to the textbooks~\cite{CyganFKLMPPS15,DowneyF13} for an introduction to the field of parameterized algorithms
and to \cite{FominLSZ19} for an introduction to kernelization.

One way to refute the existence of polynomial kernel for a parameterized problem $\Pi_2$ unless {\nka} is to obtain a special reduction from another parameterized problem $\Pi_1$  known not to 
have a polynomial kernel unless {\nka}. 


\begin{definition}[Polynomial Parameter Transformation]
\label{defn:PPT}
Let $\Pi_1, \Pi_2 \subseteq \Sigma^* \times \nn$ be two parameterized problems.
An algorithm $\cA$ is called a {\em polynomial parameter transformation} if it takes an instance $(x, k)$ of $\Pi_1$ and outputs an instance $(x', k')$ of $\Pi_2$	 such that 
\begin{itemize}
       \item $\cA$ runs in polynomial time, 
	\item $(x, k) \in \Pi_1$ if and only if $(x', k') \in \Pi_2$, and
	\item $k'$ is polynomial in $k$.
\end{itemize}
\end{definition}

\begin{proposition}~\cite{BodlaenderTY11}
\label{prop:basic-kernel-lower-bound}
Suppose that $\Pi_1$ and $\Pi_2$ be two parameterized problems such that as nonparameterized problems, $\Pi_1$ is {\rm NP-complete} and $\Pi_2$ is in {\rm NP}. 
If there is a polynomial parameter transformation from $\Pi_1$ to $\Pi_2$ and $\Pi_1$ has no polynomial kernel unless \rm{\nka}, then $\Pi_2$ has no polynomial kernel unless \rm{\nka}.
\end{proposition}
 
\subsection{Parameterized Optimization Problems and Approximate Kernels}\label{sec:approxkerneldef}
Informally, an $\alpha$-approximate kernelization is a polynomial-time algorithm that, given an instance $(I,k)$ of a parameterized problem, outputs an instance $(I',k')$ (called an {\em $\alpha$-approximate kernel}) such that $|I'| + k'\le g(k)$ for some computable function $g$, 
and any $c$-approximate solution to the instance $(I',k')$ can be turned into a $(c\alpha)$-factor approximate solution of $(I,k)$ in polynomial time. We will mainly follow Lokshtanov et al.~\cite{LokshtanovPRS17} to formally introduce approximate kernelization. (Several other recent papers studied approximate kernelization, see e.g.~\cite{EibenHR19,EibenKMPS19,LokshtanovPRS17,Ramanujan19}.)

\begin{definition}
\label{defn:para-opt}
A {\em parameterized optimization} (maximization or minimization) problem is a computable function $\Pi: \Sigma^* \times \nats \times \Sigma^* \rightarrow \mathbb{R} \cup \{\pm \infty\}$.
\end{definition}

The {\em instances} of a parameterized optimization problem $\Pi$ are pairs $(x,k) \in \Sigma^* \times \nats$, and a {\em solution} to $(x,k)$ is simply a string $s \in \Sigma^*$ such that $|s| \leq |x| + k$.
The {\em value} of the solution $s$ is $\Pi(x,k,s)$.
The problems we deal with in this paper are all minimization problems.
Therefore, we state the definitions only in terms of minimization problems (the definition of maximization problems are analogous).
Consider {\PECVC} parameterized by solution size.
This is a minimization problem where the optimization function $pECVC: \Sigma^* \times \nats \times \Sigma^* \rightarrow \mathbb{R} \cup \{\infty\}$ is defined as follows.
\[{pECVC}(G,k,S) = \left\{
\begin{array}{rl}
\infty & \text{if $S$ is not a $p$-edge-connected vertex cover of $G$},\\
\text{min}\{|S|, k + 1\} & \text{otherwise.}
\end{array}
\right.
\]
For {\PVCVC}, {the optimization function ${pVCVC}(G,k,S)$ is defined similarly by changing the first condition in the natural way.}

\begin{definition}
\label{defn:opt-value}
For a parameterized minimization problem $\Pi$, the {\em optimum value} of an instance $(x,k) \in \Sigma^* \times \nats$ is $\OPT_{\Pi}(x,k) = \min_{s \in \Sigma^*, |s| \leq |x| + k} \Pi(x,k,s)$.
\end{definition}

Naturally for the case of {\PECVC}, we denote 
$$\OPT_{\mathsf{ecvc}}(G,k) = \min_{S \subseteq V(G)} pECVC(G,k,S).$$ {The optimal objective value~$\OPT_{\mathsf{cvc}}$ of {\PVCVC} is defined analogously. Throughout the paper, we drop the subscript when this does not lead to confusion.}

\begin{definition}
\label{defn:approx-preprocess-algo}
Let $\alpha \geq 1$ be a real number and $\Pi$ a parameterized minimization problem.
An {\em $\alpha$-approximate polynomial time preprocessing algorithm $\cA$} for $\Pi$ is a pair of polynomial-time algorithms as follows.
The first one is called the {\em reduction algorithm} that computes a map $\cR_{\cA}: \Sigma^* \times \nats \rightarrow \Sigma^* \times \nats$.
Given an input instance $(x,k)$ of $\Pi$, the reduction algorithm outputs another instance $(x',k') = \cR_{\cA}(x,k)$.

The second algorithm is called the {\em solution lifting algorithm}.
This algorithm takes an input instance $(x,k)$, the output instance $(x',k')$, and a solution $s'$ to the output instance $(x',k')$.
The solution lifting algorithm works in time polynomial in $|x|,k,|x'|,k'$, and $|s'|$, and outputs a solution $s$ to $(x,k)$ such that the following holds. 
$$\frac{\Pi(x,k,s)}{\OPT(x,k)} \leq \alpha\cdot \frac{\Pi(x',k',s')}{\OPT(x',k')}.$$
The {\em size} of the reduction algorithm $\cA$ is a function $\rm{size}_{\cA}: \nats \rightarrow \nats$ defined as  ${\rm{size}_{\cA}}(k) = \sup\{|x'| + k': (x',k') = \cR_{\cA}(x,k), x \in \Sigma^*\}$.
\end{definition}

\begin{definition}
\label{defn:approx-kernel}
Let $\alpha \geq 1$ be a real number.
An {\em $\alpha$-approximate kernelization} (or an {\em $\alpha$-approximate kernel}) for a parameterized optimization problem $\Pi$, is an $\alpha$-approximate polynomial time preprocessing algorithm $\cA$ for $\Pi$ such that $\rm{size}_{\cA}$ is upper bounded by a computable function $g: \nats \rightarrow \nats$.
\end{definition}

By definition, an $\alpha$-approximate kernel for a parameterized optimization problem $\Pi$ takes input $(x, k)$ and outputs an instance $(x',k')$ of $\Pi$ in polynomial time and satisfies the other properties as given by Definition~\ref{defn:approx-preprocess-algo}.
We say that $(x',k')$ is an {\em $\alpha$-approximate polynomial kernel} if $g$ is a polynomial function (in other words if $\rm{size}_{\cA}$ is \magnus{bounded by} a polynomial function).

\begin{definition}
\label{defn:psaks}
A {\em polynomial size approximate kernelization scheme (PSAKS)} for a parameterized optimization problem $\Pi$ is a family of $\alpha$-approximate polynomial kernelization algorithms, one such algorithm for every $\alpha > 1$.
\end{definition}

\begin{definition}
\label{defn:time-efficient-psaks}
A PSAKS is said to be {\em time efficient} if both the reduction algorithm and the solution lifting algorithm run in time $f(\alpha)|x|^c$ for some function $f$ and a constant $c$ independent of $|x|,k,$ and $\alpha$.
\end{definition}

\subsection{Matroids}
\label{sec:matroid-prelims}

{We now present the preliminaries on matroids needed to obtain an FPT algorithm for {\PECVC}.}

\begin{definition}[Matroid]
\label{defn:matroid}
Let $U$ be a universe and $\cI \subseteq 2^U$.
Then, $(U, \cI)$ is said to be a {\em matroid} if the following conditions are satisfied.
\begin{enumerate}
	\item\label{matroid-basic} $\emptyset \in I$,
	\item\label{matroid-hereditary} if $A \in \cI$, then for all $A' \subseteq A$, $A' \in \cI$, and
	\item\label{matroid-exchange} if there exist $A, B \in \cI$ with $|A| < |B|$, then there exists $x \in B \setminus A$ such that $A \cup \{x\} \in \cI$.
\end{enumerate}
A set $A \in \cI$ is called an {\em independent set}.
\end{definition}

Note that all maximal independent sets of a matroid $M$ are of the same size, 
called the {\em rank of $M$} and denoted by $\rank(M)$. A maximal independent set is called a
\emph{basis} of $M$. 

We \bmp{present} some useful standard constructions. 
Let $U$ be a universe with $n$ elements and $\cI = \{A \subseteq U \mid |A| \leq r\}$.
Then, $(U, \cI)$ is a matroid called a {\em uniform matroid}.
Next, let $G = (V, E)$ be an undirected graph and let $\cI = \{F
\subseteq E(G) \mid G' = (V, F) \text{ is a forest}\}$.
Then, $(E, \cI)$ is a matroid called a {\em graphic matroid}. 
Finally, let $U$ be partitioned as $U=U_1  \cup \ldots \cup U_r$ and let $\cI=\{A \subseteq U \mid |A \cap U_i| \leq 1\, \forall i \in [r]\}$.
Then, $(U,\cI)$ is a matroid called a \emph{partition matroid}.

A matroid $M$ is said to be {\em representable over a field $\bbF$} if there is a matrix $A$ over $\bbF$ and a bijection $f \colon U \rightarrow \col(A)$, where $\col(A)$ is the set of columns of $A$, such that $B \subseteq U$ is an independent set in $M$ if and only if $\{f(b) \mid b \in B\}$ is linearly independent over $\bbF$.
Clearly the rank of $M$ is the rank of the matrix $A$.
A matroid representable over a field $\bbF$ is called a {\em linear matroid} over $\bbF$.
A graphic matroid and partition matroid can be represented over any field,
while a uniform matroid $(U, \cI)$ with $|U| = n$, can be represented over any field ${\bbF}_p$ with $p > n$. 
Furthermore, all these representations can be constructed in deterministic polynomial time.
See Oxley~\cite{OxleyBook2} for details; see also~\cite{Marx09-matroid,CyganFKLMPPS15} for expositions of the issues closer to our needs.


Given two matroids $M_1=(E_1, \mathcal{I}_1)$ and $M_2=(E_2,\mathcal{I}_2)$,
the \emph{direct sum} $M=M_1 \oplus M_2$ is the matroid $M=(E, \mathcal{I})$
where $E=E_1 \uplus E_2$ is the disjoint union of $E_1$ and $E_2$
and $I \subseteq E$ is independent in $M$ if and only if $I \cap E_1 \in \cI_1$
and $I \cap E_2 \in \cI_2$. If $M_1$ and $M_2$ are represented by matrices $A_1$
and $A_2$, respectively, over a common field $\mathbb{F}$, then a
representation of $M$ can be produced as
\[
  A =
  \begin{pmatrix}
    A_1 & 0 \\
    0 & A_2
  \end{pmatrix}.
\]
Clearly, this can be generalized to the direct sum of an arbitrary number of matroids $M_i$
with representations $A_i$ over a common field.

Given a matroid $M=(U,\cI)$, the \emph{truncation} of $M$ to rank $r$ is the matroid $M'=(U, \cI')$ where
a set $A \subseteq U$ is independent in $M'$ if and only if $A \in \cI$ and $|A| \leq r$.
Given a representation of $M$, a representation of a truncation of $M$
can be computed relatively easily in randomized polynomial time~\cite{Marx09-matroid},
but can also be computed in deterministic polynomial time through more
involved methods~\cite{LokshtanovMPS18}. For more information on matroids, see Oxley~\cite{OxleyBook2}.

\subsection{Existence of Highly Connected Vertex Covers}

We will use the following useful lemma whose vertex-connectivity part is proved in \cite{West} (Lemma 4.2.2).  We were unable to find a proof of the edge-connectivity part of the lemma in the literature and thus provide a short proof here.

\begin{lemma}\label{lem:exp}
If $G$ is a $p$-vertex-connected ($p$-edge-connected, respectively) graph and $G'$ is obtained from $G$ by adding a new vertex $y$ adjacent with at least $p$ vertices of $G$, then $G'$ is $p$-vertex-connected ($p$-edge-connected, respectively).
\end{lemma}
\begin{proof} Let $C$ be a  cut in $G'$. If $C=(y,V(G))$ then by construction $|C|\ge p$. Otherwise, let $C=(A, B)$ such that $A \uplus B = V(G')$.
Depending on whether $y \in A$ or $y \in B$, $(A \setminus \{y\}, B)$ or $(A, B \setminus \{y\})$ is a cut in $G,$ respectively.
As $G$ is $p$-edge-connected, $(A \setminus \{y\}, B)$ or $(A, B \setminus \{y\})$ has at least $p$ edges.
This means that the cut $C$ itself has at least $p$ edges.
\end{proof}

The next lemma will help us to identify whether a graph $G = (V, E)$ has a $p$-vertex/edge-connected vertex cover or not.

\begin{lemma}
\label{lemma:pvcvc-existence}
Let $G = (V, E)$ be a graph,  $L$ the set of vertices of $G$ with degree at most $p-1$ and \bmp{$S = V(G) \setminus L$}.
Then $G$ has a $p$-vertex/edge-connected vertex cover if and only if $S$ is a $p$-vertex/edge-connected vertex cover of $G$.
\end{lemma}
\begin{proof}
The backward direction $(\Leftarrow)$ of the proof is trivial.
	
We will prove the forward direction $(\Rightarrow)$.
Let $S^*$ be a $p$-vertex/edge-connected vertex cover of $G$.
If \bmp{$S^* = V(G) \setminus L$}, then we are done.
Hence, we may assume that \bmp{$S^* \ne V(G) \setminus L.$} Observe that
$S^* \cap L = \emptyset$ as otherwise for every $u \in S^* \cap L$ deleting the vertices adjacent to  $u$ (the edges incident to $u$)  will make $G[S^*]$ disconnected. {(For the vertex connectivity variant, we rely here on the fact that~$|S^*|\geq p + 1$ by definition of $p$-connectivity, so that at least one vertex besides~$u$ remains in~$S^*$ when deleting its neighbors.)} 
Let $A = V(G) \setminus (S^* \cup L)$ and note that $A$ is an independent set as $S^*$ is a vertex cover. \bmp{For each vertex~$a \in A$ of the independent set, which has degree at least~$p$, all its neighbors belong to the vertex cover~$S^*$.} 
Hence, by Lemma \ref{lem:exp} adding $A$  to $S^*$ will result in  a $p$-vertex/edge-connected vertex cover of $G$.
So, $S = S^* \cup A$ is a $p$-vertex/edge-connected vertex cover of $G$, which
completes the proof.
\end{proof}

\section{FPT Algorithms}
\label{sec:improved-pedge-cvc}

In this section, we obtain FPT algorithms for the two VC generalizations. While the algorithm  for {\PVCVC} described in Section \ref{sec:pvcvc-basic-results} is quite simple and runs in time $\cO^*(2^{\cO(k^2)})$, 
the one for {\PECVC} obtained in Section \ref{sec:singleexp} is based on matroid techniques and is of complexity $\cO^*(2^{\cO(pk)})$. 

\subsection{FPT Algorithm for {\PVCVC}}
\label{sec:pvcvc-basic-results}

\begin{theorem}
\label{thm:pvcvc-is-fpt-simple}
\bmp{For every fixed $p \geq 1$, {\PVCVC} can be solved in  $\cO^*(2^{\cO(k^2)})$ deterministic time and polynomial space.} 
\end{theorem}

\begin{proof}
Suppose that $(G, k)$ is the input instance of {\PVCVC}. \bmp{Since isolated vertices can be removed without changing the answer to the problem, we may assume without loss of generality that~$G$ does not have any isolated vertices. }
{We compute a set $H$ of vertices that have to be in any $p$-connected vertex cover of size at most $k$, as follows.
We put a vertex~$u$ into $H$ if the degree of this vertex is at least $k+1$ (so any vertex cover avoiding~$u$ is too large), or it is a neighbor to a vertex $v$ such that $\deg_G(v) \leq p - 1$ (so any vertex cover containing~$v$ is not $p$-connected).}
If $|H|>k$ then $(G, k)$ is a no-instance, so we may assume that $|H| \le k.$
Now, we partition the vertices of $G - H$ into two parts.
We put a vertex $u$ into $I$ if $N_G(u) \subseteq H$.
Otherwise, we put $u$ into $R$. Observe that $I$ is an independent set.
\bmp{If~$G[R]$ has more than~$k^2$ edges, then~$(G,k)$ is a no-instance since any vertex in $R$ has degree at most $k$, so that~$k$ vertices cover at most~$k^2$ edges of this induced subgraph. Thus in the remainder we assume~$G[R]$ has at most~$k^2$ edges. This implies~$|R| \leq 2k^2$ since each vertex~$v \in R$ has an edge to another vertex in~$R$, as by definition~$v \in R$ has a neighbor~$u$ outside~$H$ and~$u \notin I$ since~$u$ has neighbor~$v \notin H$, hence~$u\in R$.}

\bmp{We say that two vertices~$u \neq v$ are \emph{false twins} in~$G$ if~\magnus{$N_G(u)=N_G(v)$}. We will argue that when a vertex~$v \in I$ has more than~$k+1$ false twins, we can safely remove~$v$ without changing the answer: graph~$G-v$ has a $p$-connected vertex cover of size at most~$k$ if and only if~$G$ has one. For the first direction, note that any $p$-connected vertex cover~$S$ of~$G-v$ with~$|S|\leq k$ must contain all vertices of~$H$. \magnus{For a vertex $u \in H$, either $u$ has degree larger than~$k$ in~$G$, and the same holds in~$G-v$ since the vertex~$v$ we remove has~$k+1$ false twins remaining in~$G-v$; or $N_G(u)=N_{G-c}(u)$ and there is a vertex $w \in N_G(u)$ of degree less than $p$.}
Hence~$S$ contains all of~$H$, which by~$v \in I$ implies that~$S$ contains~$N_G(v)$ so that~$S$ is a vertex cover of~$G$. Since~$(G-v)[S] = G[S]$ this shows~$S$ is a $p$-connected vertex cover of~$S$.

For the converse direction, suppose~$S$ is a $p$-connected vertex cover of~$G$ of size at most~$k$, which implies~$S \supseteq H$. If~$v \notin S$, then since~$S$ is also a vertex cover in~$G-v$ and~$G[S]=(G-v)[S]$, this shows~$G-v$ has a $p$-connected vertex cover of size at most~$k$. Now suppose~$v \in S$ and let~$u_1, \ldots, u_{k+1} \in I$ be~$k+1$ false twins of~$v$ in~$G$. Since~$|S| \leq k$, there exists some~$u_i \notin S$; define~$S' := (S \setminus \{v\}) \cup \{u_i\}$. The graph~$G[S']$ is isomorphic to~$G[S]$ since~$v$ and~$u_i$ are false twins, hence~$G[S'] = (G-v)[S']$ is $p$-connected. Since all edges covered by~$v$ were already covered by~$H \subseteq S \cap S'$, the set~$S'$ is a $p$-connected vertex cover of~$G$ of size at most~$k$ that does not contain~$v$, and therefore is also a solution in~$G-v$.}

\bmp{The argumentation above shows that a vertex~$v \in I$ which has more than~$k+1$ false twins can safely be removed. After exhaustively applying this operation,}
for any $A \subseteq H$, there are at most $k + 1$ vertices in $I$ whose neighborhood equals $A$.
Hence, $|I| \leq 2^k (k+1)$. 
\bmp{It follows that after exhaustive reduction, the graph is partitioned in~$H \cup I \cup R$ with $|H| + |R| + |I| \leq k + 2k^2 + 2^k k$.}
\bmp{We now consider all subsets of $H\cup I\cup R$ with at most $k$ vertices, of which there are at most~$|H \cup I \cup R|^k \leq (k + 2k^2 + 2^k k)^k \leq 2^{\cO(k^2)}$. For each such subset~$S$, we test whether it forms a $p$-connected vertex cover in polynomial time via Menger's theorem. After trying all candidates for~$S$, we either find a $p$-connected vertex cover or conclude that the answer is no 
in time $\cO^*(2^{\cO(k^2)})$.} \bmp{Since the graph reduction can be computed in polynomial time and space, while the iteration over vertex sets of size~$k$ and the verification of a potential solution can easily be done in polynomial space, the space usage of the algorithm is polynomial in the size of the input.}
\end{proof}

\subsection{Single-Exponential Algorithm for $p$-Edge-Connected Vertex Cover}\label{sec:singleexp}

In this subsection, we provide a single-exponential algorithm for {\PECVC} using dynamic programming and the method of
\emph{representative sets}. This method was introduced to FPT algorithms by Marx~\cite{Marx09-matroid};
Fomin et al.~\cite{FominLPS16} presented additional applications and a faster method of computing such sets.
We \bmp{give} the definitions \bmp{below}.

\begin{definition}
  Let $M=(E,\mathcal{I})$ be a matroid and $X, Y \subseteq E$. 
  We say that $X$ \emph{extends} $Y$ in $M$ if $X \cap Y = \emptyset$
  and $X \cup Y \in \mathcal{I}$. Furthermore, let $\mathcal{S}$ be a
  family of subsets of $E$.  A subfamily $\hat{ \mathcal{S}} \subseteq \mathcal{S}$
  is \emph{$q$-representative} for $\mathcal{S}$ if the following holds:
  for every set $Y \subseteq E$ with $|Y| \leq q$, there is a set
  $X \in \mathcal{S}$ that extends $Y$ if and only if there is 
  a set $\hat{X} \in \hat{\mathcal{S}}$ that extends $Y$.
\end{definition}

\begin{theorem}[Fomin et al.~\cite{FominLPS16}]
  \label{thm:repset}
  Let $M=(E,{\mathcal{I}})$ be a linear matroid of rank $p+q=k$ over some field $\mathbb{F}$ 
  and let $\mathcal{S}=\{S_1, \ldots, S_t\}$ be a family of
  independent sets in $M$, each of cardinality $p$. A {$q$-representative} 
  subset $\hat{\mathcal{S}} \subseteq \mathcal{S}$ of size
  $|\hat{ \mathcal{S}}| \leq \binom{p+q}{p}$ can be computed in
  $\cO(k^{\cO(1)} \binom{p+q}{p}^{\omega-1}t)$ field operations.
  Here, $\omega < 2.37$ is the matrix multiplication exponent. 
\end{theorem}

Our algorithm is based on the following result, due to Agrawal et
al.~\cite{AgrawalMPS17}.  An \emph{out-branching} of a
digraph is a spanning subgraph which is a rooted tree, where every arc is
oriented away from the root; or equivalently, a tree where the root has no
incoming arc and every other vertex has
exactly one incoming arc. 

\begin{lemma}[Agrawal et al.~\cite{AgrawalMPS17}]
  \label{lemma:agrawal}
  Let $G=(V,E)$ be an undirected graph and let $v_r \in V$.  Define a
  digraph $D_G=(V,A_E)$ by adding the arcs $(u,v), (v,u)$ to $A_E$ for
  every edge $uv \in E$. Then $G$ is $p$-edge-connected if and only if
  $D_G$ has $p$ pairwise arc-disjoint out-branchings rooted in $v_r$. 
\end{lemma}

We will use representative sets to facilitate a fast
dynamic-programming algorithm.  We begin by observing how
out-branchings are realized using matroid tools.  Let $G=(V,E)$
and $D_G=(V,A_E)$ be as defined above and let $v_r \in V$.
The \emph{out-partition matroid (with root $v_r$)} for a digraph $D$
with $v_r \in V$ is the partition matroid with ground set $A_E$ where
arcs are partitioned according to their heads and where arcs $(u,v_r)$
are dependent.
It means that an arc set $F$ is independent in the out-partition matroid if and only if $v_r$ has in-degree 0 in $F$ and every other
vertex has in-degree at most one in $F$.
{The \emph{graphic matroid on ground set $A_E$} is the graphic matroid for $G$, where every arc $(u,v)$ represents its underlying edge $uv$ and where anti-parallel arcs $(u,v)$, $(v,u)$ represent distinct copies of $uv$.  Note that $\{(u,v),(v,u)\}$ is thus a dependent set.}
The following proposition then holds true from the characteristics of out-partition matroid and graphic matroid.


\begin{proposition} \label{prop:out-branching} $F$ is the arc set of
  an out-branching rooted in $v_r$ if and only if $|F|=|V(G)|-1$ and $F$
  is independent in both the out-partition matroid for $D_G$ with root $v_r$
  and the graphic matroid for $G$ on ground set $A_E$.
\end{proposition}

We extend this to construct a matroid that can be used to verify the
condition of Lemma~\ref{lemma:agrawal}. 

\begin{lemma} \label{lemma:our-condition}
  Let $v_r \in V(G)$ be fixed vertex. 
  Let $M$ be the direct sum of $2p+1$ matroids $M_i$ as follows.
  Matroids $M_1$, $M_3$, \ldots, $M_{2p-1}$ are copies of the graphic matroid of $G$
  {on ground set $A_E$}. 
  Matroids $M_2$, $M_4$, \ldots, $M_{2p}$ are copies of the
  out-partition matroid for $D_G$ with root $v_r$.
  Matroid $M_{2p+1}$ is the uniform matroid over $A_E$ with rank $p(k-1)$.
  Let $F \subseteq A_E$.  The following are equivalent.
  \begin{enumerate}
  \item $F$ is the arc set of $p$ pairwise arc-disjoint out-branchings
    rooted in $v_r$ in $D_G[S]$ for some $S \in \binom{V(G)}{k}$ with $v_r \in S$.
  \item $|V(F)|=k$, $|F|=p(k-1)$, $v_r \in V(F)$, and there is an
    independent set $I$ in $M$ where every arc $a \in F$ occurs in $I$
    precisely in its copies in matroids $M_{2i-1}$, $M_{2i}$ and
    $M_{2p+1}$  for some $i \in [p]$. \label{item:ind-set-from-arc-set}
  \end{enumerate}
  Furthermore, a representation of $M$ can be constructed in
  deterministic polynomial time. 
\end{lemma}

\begin{proof}
  Let $M$ and $M_i$, $i \in [2p+1]$ be as described. We note that
  graphic matroids, partition matroids (hence out-partition matroids),
  and uniform matroids all have deterministic representations~\cite{OxleyBook2}.
  Specifically, graphic matroids and partition matroids are
  representable over any field, and uniform matroids over any sufficiently large field.
  Hence a representation of $M$ can be constructed as a diagonal block
  matrix with $2p+1$ blocks, where each block $i$ is a representation of the matroid
  $M_i$ over some sufficiently large field $\mathbb{F}$ (e.g., $\mathbb{F}=\mathbb{F}_q$ for
  some prime $q>p(k-1)$).

  On the one hand, let $F$ meet the conditions in the first item.
  Since $F$ is spanning for $D_G[S]$ we have $|V(F)|=k$ and $v_r \in V(F)$, 
  and furthermore $|F|=p(k-1)$ since each out-branching is spanning.
  Furthermore, any arc set of an out-branching is independent
  in both the graphic matroid and the out-partition matroid by Prop.~\ref{prop:out-branching}.
  Letting $F=F_1 \cup \ldots \cup F_p$ where $F_i$ is the arc set of
  an out-branching for every $i \in [p]$, we can then construct $I$
  by letting an arc $a \in F_i$ be present in $I$ in matroids
  $M_{2i-1}$, $M_{2i}$ and $M_{2p+1}$. 
  Hence all conditions in the second item are met.

  Now assume that $F$ meets the conditions in the second item. 
  Partition $F=F_1 \cup \ldots \cup F_p$ where $F_i$, $i \in [p]$
  contains those arcs of $F$ that are represented in matroids
  $M_{2i-1}$ and $M_{2i}$. By a counting argument, $|F_i|=k-1$ for
  every $i \in [p]$. Furthermore, $F_i$ is the arc set of an
  out-forest (since its underlying undirected edge set is acyclic
  and every vertex has in-degree at most 1 in $F_i$).
  But then $F_i$ must form a spanning tree of $V(F)$ by counting,
  hence an out-branching of $D_G[V(F)]$ by Prop.~\ref{prop:out-branching}. Furthermore $v_r \in V(F)$ 
  and every arc into $v_r$ is dependent in $M_{2i}$; thus every
  out-branching $F_i$ is rooted in $v_r$.
\end{proof}

{To compute a representative family}, we need to modify $M$ so that the set $I$ being described is a basis of $M$, not just an
independent set.  This is more technical, but can be done
deterministically using the operation of deterministic truncation, due
to Lokshtanov et al.~\cite{LokshtanovMPS18}, as noted in Section~\ref{sec:matroid-prelims}.
We have the following lemma that ensures our matroid is linear, can be constructed in polynomial time, and satisfies the properties that we need.


\begin{lemma} \label{lemma:trunc-rep}
  Let $M'=M_1 \oplus \ldots \oplus M_{2p+1}$ be the matroid defined in Lemma~\ref{lemma:our-condition}.
  Let $M$ be the truncation of $M'$ to rank $r=3p(k-1)$.  Then a representation of $M$
  can be computed in deterministic polynomial time.  Furthermore, every set $I$
  as defined in Item~\ref{item:ind-set-from-arc-set} of Lemma~\ref{lemma:our-condition}
  is a basis of $M$. 
\end{lemma}


\begin{proof}
  By standard methods~\cite{OxleyBook2}, for every $i \in [2p+1]$ 
  we can compute a matrix $A_i$ representing the matroid $M_i$,
  where furthermore all matrices $A_i$ are over a common field $\mathbb{F}$.
  Thereby, we can also construct a matrix $A$ representing their direct sum $M'$.
  A matrix $A'$ representing the $r$-truncation of $M$ over a finite field $\mathbb{F}'$
  can then be constructed in deterministic polynomial time by 
  Lokshtanov et al.~\cite{LokshtanovMPS18}.
  For the final statement, since every set $I$ as described is independent
  in $M'$ and has $|I|=r$, every such set must be a basis of $M$. 
\end{proof}

Our algorithm for {\sc \pcvc} works as follows.
First, we enumerate all minimal vertex covers of $G$ of size at most $k$.
The number of such minimal vertex covers is at most $2^k$, and they can be enumerated in $\cO^*(2^k)$ time and space (see~\cite{MolleRR08}).
Then, for every minimal vertex cover $H$ of $G$, we use representative
sets and the above characterization to check if it can be extended to a feasible $p$-edge-connected vertex cover $S^{\star} \supseteq H$ of size at most $k$.
In detail, consider a graph $G$ with a vertex cover $H$ which is not
$p$-edge-connected, where we are looking for a set $S^{\star} \supset H$
such that $G[S^{\star}]$ is $p$-edge-connected and $|S^{\star}| \leq k$.
By iteration over $k$, we may assume that $|S^{\star}|<k$ is impossible.


Let us fix $v_r \in H$. 
By \bmp{Lemma~\ref{lemma:our-condition}}, there exists such a set $S^{\star}$ if and only if there is an independent set $I$ in $M$ meeting the following conditions:
\begin{enumerate}
\item $H \subset V(I)$ and $|V(I)|=k$.
\item $|I|=3p(k-1)$.
\item Every arc which is represented in $I$ is represented in
  precisely three matroids $M_{2i-1}$, $M_{2i}$, $M_{2p+1}$ in $M$ for some $i \in [p]$.\label{cond:three:matroids}
\end{enumerate}

The first condition can be 
reformulated as $H \subset V(I)$ and $|V(I) \setminus H|=k-|H|$. 

We can construct $I$ via dynamic programming.  The dynamic program is
set up via a table keeping track of $|V(I)|$ and $|I|$, and we ensure
that every time we add some arc $a$ to a set $I$ we add it in
precisely three layers, as described.  Thanks to the use of
representative sets, each table entry in the dynamic programming only
needs to contain $2^{\cO(pk)}$ partial solutions.

We provide the details of this scheme in the proof of the main result of this section, Theorem~\ref{thm:improved-algo-pecvc} (we restate it here).


\pECVCSinglyExponentialAlgo*

\begin{proof}
  As outlined above, we may assume that we have a vertex cover $H$
  that is not $p$-edge-connected and have already tested that there is
  no $p$-edge-connected vertex cover with at most $k-1$ vertices.
Arbitrarily order the vertices of $V(G) \setminus H$: $v_1,\ldots,v_{n'}$. Construct $M$ as in Lemma~\ref{lemma:trunc-rep}. 
  We create a dynamic programming table $T[(i,j,q)]$ with entries
  indexed by $(i,j,q)$ for $i \leq k$, $j \leq n'$ and $q \leq 3p(k-1)$. 
  Every table slot $T[(i,j,q)]$ is a collection of independent sets in $M = (E(M), \cI)$.
  Any independent set $I \in T[(i, j, q)]$ satisfies $|V(I) \setminus H|=i$, the largest-index vertex of $V(G) \setminus H$
  occurring in $V(I)$ is $v_j$, and $|I|=q$. Furthermore, for every
  independent set $I$ in the table, every arc occurring in $I$ occurs 
  in precisely three layers, as described by \bmp{Condition  \ref{cond:three:matroids}}. We may then check for a
  solution by checking whether any slot $T[(k-|H|,j,3p(k-1))]$ is
  non-empty. 

  For an arc $a \in A_E$, define $F_{a,i}$ to be the set consisting of
  the copies of $a$ in $M_{2i-1}$, $M_{2i}$ and $M_{2p+1}$. 

  We initialize the slots $T[(0,0,q)]$ by a dynamic programming
  process within $D_G[H]$. Initialise $T[(0,0,0)]=\{\emptyset\}$.
  Enumerate the arcs of $D_G[H]$ as $a_1, \ldots, a_{m'}$. 
  Then, for every $q = 3i$, where  $i \in [p(k - 1)],$
  fill in the slot $T[(0,0,q)]$
  from $T[(0,0,q-3)]$ as follows:
  \begin{enumerate}
  \item For every $I \in T[(0,0,q-3)]$, every arc $a_j$, and every $i
    \in [p]$ such that $F_{a_j,i}$ extends $I$, add $I \cup F_{a_j,i}$ to
    $T[(0,0,q)]$
  \item Reduce $T[(0,0,q)]$ to a $(3p(k-1)-q)$-representative set in
    $M$ \bmp{by applying Theorem~\ref{thm:repset}, which is justified since the rank of~$M$ is $3p(k-1)$.}
  \end{enumerate}
  This is a polynomial number of steps, where every set $I \in T[(0,0,q)]$
  is used in a polynomially bounded number of new sets. Hence every
  time we apply Theorem~\ref{thm:repset} at a level $q$, we do so with
  $t \leq (k+p)^{\cO(1)} \binom{3p(k-1)}{q}$. Thus up to polynomial
  factors each step takes time $\binom{3p(k-1)}{q}^\omega=2^{\cO(pk)}$. 

  For slots $T[(i,j,q)]$, we process vertices $v_j$ one at a time.
  The process is slightly more complex since each vertex $v_j$ can be
  incident to $\cO(k)$ arcs in $A_E$, but the principle is the same.
  We process slots $T[(i,j,q)]$ in lexicographic order by $(i,j,q)$.
  {Note that we are here processing slots in a ``forward'' direction, i.e.,
    we are using the sets in the slot $T[(i,j,q)]$ to populate slots $T[(i+1,j',q')]$,
    which come after $(i,j,q)$ in lexicographic order.}
  Before we process the sets in a slot $T[(i,j,q)]$, we reduce them to
  a $(3p(k-1) - q)$-representative set in $M$. Then we proceed as
  follows. For every independent set $I$ in $T[(i,j,q)]$ and every $j'$
  with $j < j' \leq n'$, we combine $I$ and $v_{j'}$ as follows.
  {As $v_{j'} \in V(G) \setminus H$ and $H$ is a minimal vertex cover of $G$, we have that $N_G(v_{j'}) \subseteq H$.
  Therefore, there are at most $2|H|$ arcs of $A_E$ incident on $v_{j'}$ in $D_G$.}
  \begin{enumerate}
  \item Let $d \leq 2|H|$ be the number of arcs of $A_E$ incident with
    $v_{j'}$. Create a set $F$ for every one out of the
    {following options: for every arc $a$ incident with $v_{j'}$,
      either add $F_{a,b}$ to $F$ for some $b \in [p]$, or do not add any set $F_{a,b}$ to $F$. 
      Note that this makes $(p+1)^d$ different sets $F$ in total. }
  \item For every such non-empty set $F$, and every independent set
    $I$ of $T[(i, j, q)]$ such that $F$ extends $I$ in $M$, add $I \cup F$
    to $T[(i+1,j',q+|F|)]$. 
  \end{enumerate}
  To roughly bound the running time, we note that the number of slots
  in the table is polynomial, and for every set $I$ in a slot, at most
  $(p+1)^{2|H|}$ sets $I'=I \cup F$ are added to other slots of the table.
  Furthermore, after the representative set reduction, every slot
  contains at most $2^{3p(k-1)}$ sets. Thus every time we apply
  Theorem~\ref{thm:repset}, we have
  \[
    |T[(i,j,q)]| < (p+k)^{\cO(1)} 2^{3pk}(p+1)^{2k} = 2^{\cO(pk)},
  \]
  hence the total time usage, up to a polynomial factor, is $2^{\cO(pk)}$.

  It remains to prove correctness.  Let $I \in T[(i,j,q)]$ for some
  $(i,j,q)$ and let $S=V(I) \setminus H$.  We note the invariants
  $|S|=i$, $\max_a \{v_a \in S\}=j$ and $|I|=q$ hold by induction.
  Furthermore, by construction $I$ is independent in $M$. 
  With these observations, we proceed.  First assume that there is a
  set $I \in T[(k-|H|,j,3p(k-1))]$ for some $j$.
  Then $|V(I) \setminus H|=k-|H|$ and $|I|=3p(k-1)$ by the invariants. 
  Furthermore, every arc $a$ represented in $I$ occurs in precisely
  three copies by construction. Indeed, every time we grow a set $I$
  we do so by adding a collection of sets $F_{a,i}$ to it. Hence every
  arc occurs at least three times, and furthermore, since $F_{a,i}$
  always contains a copy of $a$ in $M_{2p+1}$, we will never add two
  distinct sets $F_{a,i}$, $F_{a,i'}$ to the same set $I$. Hence
  Lemma~\ref{lemma:our-condition} implies that $D_G[H \cup V(I)]$ has
  $p$ pairwise arc-disjoint out-branchings rooted in $v_r$, 
  which by Lemma~\ref{lemma:agrawal} implies that $G[H \cup V(I)]$ is
  $p$-edge-connected.

  On the other hand, assume that $G[H \cup S]$ is $p$-edge-connected
  for some $S \subseteq V(G) \setminus H$ with $|S|=k-|H|$.
  By Lemma~\ref{lemma:agrawal} there exist $p$ pairwise edge-disjoint
  out-branchings in $D_G[H \cup S]$ rooted in $v_r$, hence by
  Lemma~\ref{lemma:our-condition} there is an independent set $I$ in $M$
  formed using a set of arcs $F \subseteq A_E$ with $|F|=p(k-1)$
  and $|I|=3p(k-1)$ as described.  Let $j$ be the largest index
  such that $v_j \in V(I)$; then $I$ is a candidate for the table
  slot $T[(k-|H|,j,3p(k-1))]$.  We prove by induction
  that $T[(k-|H|,j,3p(k-1))]$ is non-empty.  Observe that $I$ is the
  disjoint union of sets $F_{a,i}$.  We partition $I$ according to
  lexicographical order of $(i,j,q)$ as follows. First, let
  $F_{a_1,i_1}$, \ldots, $F_{a_t,i_t}$ enumerate the sets $F_{a,i}$
  contained in $I$ for which $a$ is contained in $D_G[H]$.
  For $r \in [t]$, let
  \[
    I_r'=I \setminus \bigcup_{j=1}^r F_{a_j,i_j}
  \]
  be the subset of $I$ which is encountered ``after'' $F_{a_r,i_r}$
  in the natural ordering. 
  We show by induction that for each $r \in [t]$, the slot
  $T[(0,0,3r)]$ contains a set which extends $I_r'$. 
  For $r=0$ this holds trivially.
  Hence, assume the statement holds for $T[(0,0,3r)]$ for some $r<t$,
  and let $I_0 \in T[(0,0,3r)]$ be a set which extends $I_r'$.
  While processing $(0,0,3r)$, the set $F_{a_{r+1},i_{r+1}}$ is
  considered in the loop, and clearly it extends $I_0$ since
  $F_{a_{r+1},i_{r+1}} \subseteq I_r'$. Hence $T[(0,0,3r+3)]$ contains
  the set $I_1=I_0 \cup F_{a_{r+1},i_{r+1}}$ before the representative
  set computation is performed.  By assumption $I_1$ extends $I_{r+1}'$.
  Hence by the correctness of Theorem~\ref{thm:repset},
  $T[(0,0,3r+3)]$ contains some set $I_2$ that extends $I_{r+1}'$, as
  required. Hence the claim holds up to the set $T[(0,0,3t)]$.

  We can now complete the proof using the same outline for entries
  $T[(i,j_i,q_i)]$. Enumerate $S$ as $S=\{v_{j_1}, \ldots, v_{j_{k-|H|}}\}$ 
  in increasing order of indices $j_i$ and for each $i \in [k-|H|]$
  let $I_i$ be the union of sets $F_{a,i}$ of $I$ for which $a$ is
  incident with $v_{j_i}$. Let $I_{\geq i}=\bigcup_{j=i}^{k-|H|} I_j$.
  We show by induction in lexicographical order that $T[(i,j_i,q_i)]$
  for some $q_i$ contains a set which extends $I_{\geq i+1}$, for each $i$.
  As a base case, the claim holds for $T[(0,0,3t)]$ as has already
  been shown. For the inductive step, the proof is precisely as above.
  For every $i=1, \ldots, k-|H|$, let $I_{i-1}' \in T[(i-1,j_{i-1},q_{i-1})]$ 
  be a set which extends $I_{\geq i}$. Then in particular $I_i$
  extends $I_{i-1}'$ and is added to table $T[(i,j_i,q_i)]$
  in the exhaustive enumeration loop from $T[(i-1,j_{i-1},q_{i-1})]$. 
  Thus before the call to Theorem~\ref{thm:repset} there was a set in
  $T[(i,j_i,q_i)]$ which extends $I_{\geq i+1}'$, hence the same holds
  after the representative set reduction. By induction, the table slot
  $T[(k-|H|,j,3p(k-1)]$ is indeed non-empty for some $j$.
  This completes the proof of the theorem.
\end{proof}

\section{Approximate Kernels}
\label{sec:psaks}

In this section we describe the approximate kernels for the two types of connectivity. \bmp{Both cases rely on a common subroutine ${\sf Mark}(G,k,\eps)$, which we present below. It works for a fixed value of~$p$. It gets as input a graph~$G$, integer~$k$, and a value of~$\eps > 0$, and works as follows.}
\begin{enumerate}
	\item Let~$H$ be the vertices in~$G$ whose degree is larger than~$k$. Let~$I$ consist of the vertices~$v$ in~$V(G) \setminus H$ with~$N_G(v) \subseteq H$. Note that~$I$ is an independent set. Let~$R = V(G) \setminus (H \cup I)$, and note that each vertex in~$R$ has degree at least one and at most~$k$.
	\item If~$|R| > 2k^2$ or~$|H| > k$ then return \textsc{infeasible}.
	\item Otherwise, we mark a set~$L \subseteq I$. Initialize~$L = \emptyset$. For each set~$S \in \binom{H}{\leq 2\lceil p/\eps \rceil}$, we mark common neighbors of~$S$ in~$I$ as follows:
	\begin{enumerate}
		\item If~$|\bigcap _{v \in S} N_G(v) \cap I| \leq (3+\eps)k$, then add all vertices of~$\bigcap _{v \in S} N_G(v) \cap I$ to~$L$.
		\item Otherwise, let~$L_S$ consist of~$\lfloor (3+\eps)k \rfloor$ arbitrary vertices from~$\bigcap _{v \in S} N_G(v) \cap I$ and add~$L_S$ to~$L$.
	\end{enumerate}
	\item Return the graph~$G' = G[H \cup R \cup L]$ with parameter value~$k' = \lceil (1+\eps)k \rceil$.
\end{enumerate}

Note that since~$G'$ is an induced subgraph of~$G$, any $p$-vertex/edge-connected vertex cover~$X$ in~$G$ with~$X \subseteq V(G')$ is also a $p$-vertex/edge-connected vertex cover in~$G'$. Note that a run which outputs a graph~$G'$ results in~$|L| \leq \lfloor (3+\eps)k \rfloor \cdot |\binom{H}{\leq 2\lceil p/\eps \rceil}| \leq \lfloor (3+\eps)k \rfloor k^{2\lceil p/\eps \rceil}$.

\bmp{The following lemma encapsulates which information ${\sf Mark}$ preserves.}

\begin{lemma} \label{lemma:lossykernel:blowup}
\bmp{If ${\sf Mark}(G,k,\eps)$ is executed \bmp{with~$0 < \eps \leq 1$} on a graph~$G$ which has a $p$-connected (resp.~$p$-edge-connected) vertex cover~$X$ of size at most~$k$, then it outputs a graph~$G[H \cup R \cup L]$ (rather than \textsc{infeasible}) and~$G[H \cup R \cup L]$ has a $p$-connected (resp.~$p$-edge-connected) vertex cover of size at most~$(1+\eps)|X|$.}
\end{lemma}
\begin{proof}
\bmp{We first argue that ${\sf Mark}$ does not output~\textsc{infeasible} when given a graph~$G$ with a $p$-connected ($p$-edge-connected) vertex cover of size at most~$k$, as follows. We have~$|H| \leq k$ as each vertex of~$H$ has degree more than~$k$ and belongs to each vertex cover of size at most~$k$. Similarly, we must have~$|R| \leq 2k^2$: as~$R$ consists of vertices outside~$I \cup H$, each vertex of~$R$ has a neighbor that does not belong to~$H$ and which is therefore in~$R$ itself. Hence~$G[R]$ contains no isolated vertices and has maximum degree~$k$. If~$|R| > 2k^2$, then it has more than~$k^2$ edges which cannot be covered using at most~$k$ of its vertices. Hence if~$|R| > 2k^2$ then~$G[R]$ has no vertex cover of size at most~$k$ and neither does~$G$.}

\bmp{It follows that under the stated assumption, the algorithm outputs a graph~$G' = G[H \cup R \cup L]$. \bmp{In the remainder, we refer to a $p$-vertex/edge-connected vertex cover in~$G$ (depending on the problem variant considered) as a \emph{solution}.} Let~$X$ be a solution in~$G$ of size at most~$k$. We prove that~$G[H \cup R \cup L]$ has a solution of size at most~$(1+\eps)|X|$.}

Since each vertex of~$H$ has degree more than~$k$, it follows that~$H \subseteq X$. If~$X \subseteq V(G')$ then~$X$ is also a solution in~$G'$ and the lemma follows. In the remainder, we treat the case that~$X \setminus V(G') = \{x_1, \ldots, x_\ell\}$ for some~$1 \leq \ell \leq |X| \leq k$. As~$H \cup R \cup I$ is a partition of~$V(G)$, since the output graph~$G'$ contains all of~$H$ and~$R$ it follows that for each~$x_i \in X \setminus V(G')$ we have~$x_i \in I$, which implies~$N_G(x_i) \subseteq H \subseteq X$. We will show that we can replace each vertex~$x_i$ by a small set of marked vertices to obtain an approximate solution in~$G'$. 

By Proposition~\ref{thm:connectivity-certificate}, graph $G[X]$ has a $p$-vertex/edge-connected spanning subgraph $F$ on at most $p|X|$ edges. In the following argument, the degree of vertices~$x_i$ in the subgraph~$F$ will play an important role. As~$I$ is an independent set in~$G$, it is also an independent set in its subgraph~$F$. It follows that in the sum~$\sum _{i=1}^\ell \deg_F(x_i)$ we never count the same edge twice, so that the sum is bounded by the total number of edges in~$F$, which is \magnus{at most}~$p|X|$. We record this property for further use:
	
	\begin{equation} \label{eq:kernel:summation}
	\sum _{i \in [\ell]} \deg_F(x_i) \leq p|X|.
	\end{equation}

\paragraph{Constructing replacement sets} For the replacement, we construct a sequence of sets~$X'_1, \ldots, X'_\ell$. \bmp{Each set~$X'_i$ will be used as a replacement for the corresponding vertex~$x_i$. The sets we construct will have} the following properties:

\begin{enumerate}[(a)]
	\item Each set~$X'_i$ is a subset of~$V(G') \cap I$ that is disjoint from~$X \cup \bigcup_{j < i} X'_j$.
	\item The vertices in each set~$X'_i$ can be ordered so that each successive pair of vertices of~$X'_i$ has at least~$p$ common neighbors in the set~$H$.
	\item $N_F(x_i) \subseteq \bigcup _{u \in X'_i} N_G(u)$.
	\item $|X'_i| \leq \max(\frac{\deg_F(x_i)}{\lceil p/\eps \rceil}, 1)$.
\end{enumerate}

{In the remainder of the proof, whenever we refer to successive or consecutive vertices of~$X'_i$, we mean with respect to the ordering whose existence is guaranteed by the second condition.} We construct the sets~$X'_i$ in order of increasing~$i$. Consider a vertex~$x_i \in X \setminus V(G') \subseteq I$. Let~$S = N_F(x_i)$ and note that~$|S| = \deg_F(x_i)$. Since~$F$ is $p$-vertex/edge-connected, we have~$|S| \geq p$. We define a partition of~$S$ as follows. If~$|S| \leq \lceil p/\eps\rceil$ then we use the singleton partition of~$S = S_1$. If~$|S| > \lceil p/\eps \rceil$ then we partition~$S$ into sets~$S_1, \ldots, S_r$ of size exactly~$\lceil p/\eps \rceil$, except for the last set which has size at least~$\lceil p/\eps \rceil$ and less than~$2 \lceil p/\eps \rceil$. Such a partition always exists. Note that~$r \leq \max(\frac{\deg_F(x_i)}{\lceil p/\eps \rceil},1)$, even without rounding, where the maximum is needed to deal with the case~$r=1$. For each~$i \in [r-1]$, let~$T_i$ be an arbitrary subset of~$S_{i+1}$ of size exactly~$p$, which exists since~$|S_{i+1}| \geq \lceil p/\eps\rceil \geq p$ since~$\eps \leq 1$. Let~$T_r = \emptyset$.

Note that~$|S_j \cup T_j| \leq 2\lceil p/\eps \rceil$ for each~$j \in [r]$: for~$j < r$ we have~$|S_j| \leq \lceil p/\eps\rceil$ and~$|T_j| \leq p \leq \lceil p/\eps\rceil$ (we use~$\eps \leq 1$ here), while the case~$j=r$ holds since~$T_r = \emptyset$. Each set~$S_j \cup T_j$ consists of neighbors of~$x_i$ in~$F$ and therefore in~$G$, which shows that~$x_i \in \bigcap_{v \in S_j \cup T_j} N_G(v) \cap I$. Hence~$x_i$ was eligible to be marked for the set~$S_j \cup T_j$, but it was not. Hence we marked a set~$L_{S_j \cup T_j} \subseteq V(G') \cap I$ of~$\lfloor (3+\eps)k \rfloor$ vertices. To show that there exist sufficiently many marked vertices which are not contained in~$X$ or in a set~$X'_{j'}$ for~$j' < j$, we bound the latter as follows:
\begin{align*}
\left |\bigcup _{j' < j} X'_{j'} \right | &\leq \sum _{1 \leq j' < j} \max \left(\frac{\deg_F(x_{j'})}{\lceil p/\eps \rceil}, 1 \right) & \mbox{By the fourth condition} \\
&\leq \sum _{j' \in [\ell]} \left (\frac{\deg_F(x_{j'})}{\lceil p/\eps \rceil} + 1 \right) & \\
&\leq \frac{\sum _{j' \in [\ell]} \deg_F(x_{j'})}{\lceil p/\eps \rceil} + \ell \leq \frac{p|X|}{\lceil p/\eps \rceil} + \ell & \mbox{By \eqref{eq:kernel:summation}}\\
&\leq k + \eps k. & \mbox{Since~$\ell, |X| \leq k$} \\
\end{align*}
So the set~$X$ contains at most~$k$ vertices of~$L_{S_j \cup T_j}$, while~$\bigcup _{j' < j} X'_{j'}$ contains at most~$k + \eps k$ vertices of~$L_{S_j \cup T_j}$. It follows that there are at least~$k$ vertices of~$L_{S_j \cup T_j}$ which belong neither to~$X$ nor to sets~$X'_i$ we already constructed. As this holds for each of the~$r$ sets into which we partitioned~$S$, while~$r \leq \deg_F(x_i) \leq \deg_G(x_i) \leq |X| \leq k$, there exist \emph{distinct} vertices~$u_1, \ldots, u_r \in V(G') \setminus (X \cup \bigcup_{j' \leq i} X'_{j'})$ such that for each~$j \in [r]$ we have~\magnus{$u_j \in L_{S_j \cup T_j}$}, which implies that~$N_G(u_j) \supseteq S_j \cup T_j$. We set~$X'_i = \{u_1, \ldots, u_r\}$. To see that this satisfies all four conditions mentioned above, observe that we get the first by construction. The second follows from the fact that for~$j \in [r-1]$, both~$u_j$ and~$u_{j+1}$ are adjacent to~$T_j \subseteq H$. Since the sets~$S_j$ partition~$S = N_F(x_i)$, we satisfy the third condition. The fourth follows from the given bound on~$r$. This completes the construction of the sets~$X'_i$.

Using these sets we complete the proof. Let~$X' = (X \setminus \{x_1, \ldots, x_\ell\}) \cup \bigcup _{i \in [\ell]} X'_i$. By \eqref{eq:kernel:summation} and the fourth condition, we can infer that~$|X'| \leq (|X| - \ell) + (\frac{p|X|}{\lceil p/\eps \rceil} + \ell) \leq (1+\eps)|X|$. Hence~$X'$ is a vertex set in~$G'$ of the appropriate size. It remains to analyze its connectivity. We split the proof into two cases here, based on whether we are considering vertex or edge connectivity. Note that~$G[X'] = G'[X']$ since~$G'$ is an induced subgraph of~$G$ with~$X' \subseteq V(G')$.

\subparagraph{Edge connectivity.} In the case of edge connectivity, the subgraph~$F$ of~$G[X]$ we chose above is $p$-edge-connected. We will argue that~$G[X']$ is $p$-edge-connected. Assume for a contradiction that~$G[X']$ has a cut~$(A', B')$ of less than~$p$ edges. For each set~$X'_i$ inserted into~$X'$, the members of~$X'_i$ can be ordered so that successive vertices have at least~$p$ common neighbors in~$H$. As~$H$ belongs to~$X$ and~$V(G')$ and therefore to~$X'$, successive vertices of~$X'_i$ have at least~$p$ common neighbors in~$G[X']$, and therefore belong to the same side of any cut of less than~$p$ edges. Hence for each set~$X'_i$ inserted into~$X'$, we have~$X'_i \subseteq A'$ or~$X'_i \subseteq B'$. This allows us to transform~$(A',B')$ into a cut~$(A,B)$ of the subgraph~$F$ on vertex set~$X$ in the natural way, by replacing each set~$X'_i$ by the corresponding vertex~$x_i$. The key observation is now that this transformation does not increase the number of edges in the cut: for each edge in the cut~$(A,B)$, either it is an edge between two vertices of~$X' \cap X$ (and therefore also an edge of the cut~$(A',B')$), or it is an edge incident on some vertex~$x_i \in I$ whose other endpoint~$v$ therefore belongs to~$H$. But then the vertex set~$X'_i$ is in the same side of the cut in~$(A',B')$ and contains a vertex adjacent to~$v$, as~$N_F(x_i) \subseteq \bigcup _{u \in X'_i} N_G(u)$. Hence the size of cut~$(A,B)$ of~$F$ is not larger than the cut~$(A',B')$ of~$G[X']$, which contradicts that~$F$ is $p$-edge-connected. Hence~$G[X']$ is~$p$-edge-connected.

\subparagraph{Vertex connectivity} In the case of vertex connectivity, the subgraph~$F$ of~$G[X]$ is $p$-connected. We argue that~$G[X']$ is also $p$-connected. Since~$|X'| \geq |X| \geq p + 1$, it suffices to verify that~$G[X']$ cannot be disconnected by removing less than~$p$ vertices. Consider a vertex set~$Z'$ of size less than~$p$; we will argue that~$G[X'] - Z'$ is connected. Let~$\mathcal{I}$ consist of those indices~$i \in [\ell]$ such that~$X'_i \cap Z' \neq \emptyset$, and let~$Z = (Z' \setminus \bigcup_{i \in \mathcal{I}} X'_i) \cup \{x_i \mid i \in \mathcal{I}\}$ be obtained by replacing each set~$X'_i$ intersecting~$Z'$ by the single vertex~$x_i$. \bmp{Since the sets~$X'_i$ are pairwise disjoint by the first condition,}~$|Z| \leq |Z'| < p$ and therefore~$F-Z$ is connected. Let~$X'' = X' \setminus \bigcup _{i \in \mathcal{I}} X'_i$. We shall first prove that~$G[X''] - Z'$ is connected, and later show how this implies connectivity of~$G[X'] - Z'$ itself. 

Assume for a contradiction that~$G[X''] - Z'$ is not connected. Consider a pair of vertices~$u,v$ that belong to different connected components of~$G[X''] - Z'$. Each vertex of each set~$X'_i$ has at least~$p$ neighbors in~$H$: if~$|X'_i| > 1$ this follows from the second condition on the subsets, while for~$|X'_i| = 1$ the third condition implies the single vertex in~$X'_i$ has at least~$\deg_F(x_i)$ neighbors in~$H$, while~$\deg_F(x_i) \geq p$ due to $p$-connectivity of~$F$. Hence any connected component of~$G[X''] - Z'$ that contains a vertex of a set~$X'_i$ for~$i \in [\ell]$, also contains a vertex of~$H$. It follows that there are two vertices~$u,v$ of~$G[X'']-Z'$ belonging to different connected components and~$u,v \notin \bigcup _{i \in [\ell]} X'_i$, so that~$u,v \in X \cap X'$. Since the process of turning~$Z'$ into~$Z$ only affected vertices outside~$X \cap X'$, vertices~$u$ and~$v$ exist in~$F - Z$ and are connected by a path~$P$ in~$F - Z$ since~$F$ is $p$-connected. We transform~$P$ into a path connecting~$u$ and~$v$ in~$G[X''] - Z'$, as follows. For each occurrence of a vertex~$x_i \in X \setminus X'$ on path~$P$, we know~$x_i \notin Z$ so~$X'_i \cap Z' = \emptyset$. Each pair of successive vertices from~$X'_i$ has~$p$ common neighbors in~$H \subseteq X''$, of which at most~$p-1$ belong to~$Z'$, so each pair of successive vertices from~$X'_i$ is connected in~$G[X''] - Z'$; {hence all vertices of~$X'_i$ belong to the same connected component of~$G[X''] - Z'$.} By the third condition on the subsets~$X'_i$, some vertex of~$X'_i$ is adjacent to the predecessor of~$x_i$ on~$P$, and some vertex of~$X'_i$ is adjacent to the successor of~$x_i$ on~$P$. Hence each occurrence of a vertex~$x_i \in X \setminus X'$ on path~$P$ can be replaced by a path through~$G[X''] - Z'$. This transforms~$P$ into a path~$P'$ connecting~$u$ and~$v$ in~$G[X''] - Z'$; a contradiction. Hence~$G[X''] - Z'$ is connected.

From the fact that~$G[X''] - Z'$ is connected, we derive that~$G[X'] - Z'$ is connected as follows. Each vertex of~$X' \setminus X''$ belongs to some set~$X'_i$ for~$i \in \mathcal{I}$. As observed above, each vertex of~$X'_i$ has at least~$p$ neighbors in~$H \subseteq X''$, of which at most~$p-1$ belong to~$Z'$; hence each vertex of~$X' \setminus X''$ is adjacent to a vertex of~$G[X''] - Z$. Therefore~$G[X'] - Z'$ can be obtained from~$G[X'']-Z'$ by inserting non-isolated vertices, which leaves the graph connected. This completes the proof that~$G[X']$ is $p$-connected.

As the above two cases show that~$G[X'] = G'[X']$ is $p$-vertex/edge-connected, while we already derived~$|X'| \leq (1+\eps)|X|$, this proves that~$X'$ is a $p$-vertex/edge-connected vertex cover of~$G$ of the appropriate size. As~$X' \subseteq V(G')$ and~$G'$ is an induced subgraph of~$G$, it is also a valid solution in~$G'$, which completes the proof.
\end{proof}

Using the previous lemma we now prove the existence of approximate kernels for the two considered problems.

\pKernelMainTheorem*

\begin{proof}
\bmp{Fix~$p \geq 2$. Before presenting the main argument, we show that to prove the theorem it suffices to prove it for~$0 < \epsilon \leq 1$. Let~$s(k,\epsilon) = k + 2k^2 + \lfloor (3+\eps)k \rfloor k^{2\lceil p/ \min(\eps,1) \rceil}$ denote the guarantee on the number of vertices in a reduced instance claimed by the theorem for a certain value of~$\epsilon$ and parameter value~$k$. For any~$k \in \mathbb{N}$ and~$\epsilon > 1$ we have~$s(k,1) \leq s(k,\epsilon)$; we rely here on the~$\min(\epsilon,1)$ term in the exponent, which means that the exponent stops becoming smaller when~$\epsilon$ becomes larger than~$1$. Suppose that the theorem holds for~$\epsilon=1$, which means that for each type of connectivity considered there is a polynomial-time reduction algorithm~$\cR_1$ reducing any instance~$(G,k)$ to an instance~$(G',k')$ on at most~$s(k,1)$ vertices, and a polynomial-time solution lifting algorithm~$\mathcal{L}_1$ that can lift $\alpha$-approximate solutions for~$(G',k')$ to $(\alpha \cdot (1+\eps)) = (\alpha \cdot (1+1))$-approximate solutions in~$(G,k)$. The algorithms~$\cR_1$ and~$\mathcal{L}_1$ also form a valid~$(1+\eps)$-approximate kernel for any~$\eps > 1$: as just argued, the output of~$\cR_1$ has at most~$s(k,1)\leq s(k,\epsilon)$ vertices, while the solution lifting algorithm produces a solution whose approximation factor is~$(\alpha \cdot (1+1)) \leq (\alpha \cdot (1+\eps))$. Hence to prove the theorem it suffices to prove it for~$0 < \eps \leq 1$.

Consider~$0 < \eps \leq 1$.} The approximate kernelization algorithm has two parts. The first part is a reduction algorithm, and the second part is a solution lifting algorithm. Let $(G, k)$ be an input instance of {\PECVC} or {\PVCVC}.

\subparagraph{Reduction algorithm} First we invoke Lemma~\ref{lemma:pvcvc-existence} to check whether~$G$ has a $p$-vertex/edge-connected vertex cover. If not, then the reduction algorithm outputs the instance~$(2K_2,1)$, that is, a matching of two edges with a parameter value of one. 

If~$|V(G)| \leq k^{2\lceil p/\eps \rceil}$, then the instance is already small in terms of the parameter. \bmp{To ensure the running time of the reduction algorithm is bounded by a polynomial whose degree does not depend on~$\epsilon$, in this case} we simply output~$(G,k)$ unchanged. If~$|V(G)|$ is larger, we run ${\sf Mark}(G,k,\eps)$. If it outputs \textsc{infeasible}, we output~$(K_{2p}, 1)$, that is, a clique of size~$2p$ with a parameter value of one. If ${\sf Mark}$ outputs an instance~$(G' = G[H \cup R \cup L], k' = (1+\eps)k)$, we use~$(G',k')$ as the output of the reduction algorithm. 

As the for-loop of ${\sf Mark}$ only happens when~$|H| \leq k$, the running time of the algorithm can be bounded as~$n^{\cO(1)} \cdot k^{2\lceil p/\eps \rceil}$ as it spends~$n^{\cO(1)}$ time for each subset of~$H$ of size at most~$2\lceil p/\eps \rceil$ while~$|H| \leq k$. By our assumption on~$|V(G)|$, we have~$k^{2\lceil p/\eps \rceil} \leq |V(G)|$ so that the running time is~$n^{\cO(1)}$ for some absolute constant not depending \bmp{on~$p$,~$k$ or~$\eps$}. Hence the approximate kernelization scheme is time-efficient. If the output is not equal to the result of ${\sf Mark}$, its size is trivially bounded as required. The output of ${\sf Mark}$ is~$G[H \cup R \cup L]$, where~$|H| \leq k$ and~$|R| \leq 2k^2$ follow from the definition of the algorithm, while~$|L| \leq \lfloor(3+\eps) k \rfloor k^{2\lceil p/\eps \rceil}$ as observed below its presentation. Hence the number of vertices in the output graph is as claimed. Since~$p$ is a fixed constant, this is suitable for a polynomial-sized approximate kernelization scheme. 

\subparagraph{Solution lifting algorithm} Given a solution~$S' \subseteq V(G')$ for the instance~$(G',k')$, the solution lifting algorithm proceeds as follows. {If~$G$ does not have a $p$-edge/vertex-connected vertex cover, it outputs~$\emptyset$ as the solution. If~$S'$ is not a valid solution in~$G'$, or~$S'$ does not contain all vertices of~$H$, then we output the trivial $p$-edge/vertex-connected vertex cover of~$G$ found via Lemma~\ref{lemma:pvcvc-existence} as the solution to~$(G,k)$. Otherwise, we output~$S'$ as the solution~$S$ for~$(G,k)$; we argue below that it is a valid solution.}

It remains to argue that the output~$S$ of the solution lifting algorithm is of sufficient quality. Formally, we need to establish that:

\begin{equation} \label{eq:kernel:quality}
\frac{{pCVC(G,k,S)}}{\OPT(G,k)} \leq (1+\eps)\frac{{pCVC(G',k',S')}}{\OPT(G',k')}.
\end{equation}

\bmp{Depending on the type of connectivity considered, ${pCVC}$ corresponds to either ${pECVC}$ or ${pVCVC}$, the functions defined below Definition~\ref{defn:para-opt} that map solutions of the considered parameterized optimization problems to their cost value.} If~$G$ does not have any $p$-vertex/edge-connected vertex cover, then the value of each solution is~$+\infty$ so each solution is optimal, which implies that \eqref{eq:kernel:quality} holds since the left-hand side becomes~$1$ and the right-hand side is never smaller. Similarly, if~$G$ has a $p$-vertex/edge-connected vertex cover, but not of size at most~$k$, then~$\OPT(G,k) = k+1$ and by definition of the function ${pCVC}$, each solution has cost at most~$k+1$ and is therefore optimal. Hence it remains to consider the case that~$G$ has a $p$-vertex/edge-connected vertex cover of size at most~$k$. 


\bmp{By Lemma~\ref{lemma:lossykernel:blowup}, the fact that~$G$ has a $p$-vertex/edge-connected vertex cover of size at most~$k$ implies that  the reduction algorithm outputs a nontrivial graph~$G'$} which has a  $p$-vertex/edge-connected vertex cover of size at most~$(1+\eps)|X|$. This implies that~$\OPT(G,k) = |X|$ and~$\OPT(G',k') \leq (1+\eps)|X|$, so that~$\OPT(G',k') \leq (1+\eps)\OPT(G,k)$. To analyze the result of the \bmp{solution} lifting algorithm, we consider two cases depending on the structure of the solution~$S'$ given to the algorithm. 

\begin{itemize}
	\item Suppose~$S'$ is a $p$-vertex/edge-connected vertex cover of~$G'$ of size at most~$k$. We argue that~$S'$ contains all vertices of~$H$: each vertex of~$H$ has degree more than~$k$ in~$G$, and for each vertex~$v \in I$ that is not marked by the algorithm and therefore no longer occurs in~$G'$, we marked~$\lfloor (3+\eps)k \rfloor > k$ vertices for each neighbor of~$v$ in~$H$. Hence each vertex of~$H$ also has degree more than~$k$ in~$G'$, which means it is contained in each vertex cover of size at most~$k$. Hence~$H \subseteq S'$. Each vertex of~\magnus{$V(G) \setminus V(G')$} belongs to the independent set~$I$ and has all its neighbors in~$H$. Hence~$S' \supseteq H$ covers all edges incident on vertices of~\magnus{$V(G) \setminus V(G')$}. As~$G[S'] = G'[S']$, this implies that~$S'$ is a $p$-vertex/edge-connected vertex cover of~$G$, which is a valid output for the solution lifting algorithm. This satisfies Equation~\ref{eq:kernel:quality} since $\OPT(G',k') \leq (1+\eps)\OPT(G,k)$.
	\item Now suppose~$S'$ is a $p$-vertex/edge-connected vertex cover of~$G'$ of size more than~$k$, which implies~${pCVC(G',k',S')} \geq k+1$. Since the solution lifting algorithm outputs a valid $p$-vertex/edge-connected vertex cover of~$G$ whenever there is one, we have~${pCVC(G,k,S)} \leq \min(|S|, k+1) \leq k+1$. Now we derive:
		\begin{equation*}
\frac{{pCVC(G,k,S)}}{\OPT(G,k)} \leq \frac{k+1}{\OPT(G,k)} \leq (1+\eps)\frac{k+1}{\OPT(G',k')} \leq (1+\eps)\frac{{pCVC(G',k',S')}}{\OPT(G',k')},
	\end{equation*}
	where the middle inequality follows from $\OPT(G',k') \leq (1+\eps)\OPT(G,k)$. 
\end{itemize}

This concludes the proof.
\end{proof}

\section{Constant Factor Approximation Algorithm for {\PECVC}}
\label{sec:approx-p-edge-cvc}

In this section, we describe a $2(p+1)$-approximation algorithm for {\PECVC}.
We begin by \bmp{defining} the notion of a \emph{Gomory-Hu tree}.

\begin{definition}[Gomory-Hu Tree]
\label{defn:gomory-hu-tree}
Let $G = (V, E)$ be a graph, and  let $c(u,v)\ge 0$ be the  {\em capacity} of edge $uv\in E,$ letting $c(u,v)=0$ if $uv \notin E.$
Denote the minimum capacity of an $s$-$t$ cut by $\lambda_{st}$ for each $s, t \in V(G)$.
Let $T = (V_T,E_T)$ be a tree with $V_T = V(G)$, and let us denote the set of edges in the $s$-$t$ path in $T$ by $P_{st}$ for each $s,t \in  V_T$.
Then $T$ is said to be a {\em Gomory-Hu tree} of $G$ if $\lambda_{st} = \min\limits_{e \in P_{st}} c(S_e, T_e)$ for all $s, t \in V(G)$,
where
\begin{itemize}
    \item $S_e$ and $T_e$ are sets of vertices of the two connected components of $T-e$ such that $s\in S_e$ and $t\in T_e,$ and
    \item $c(S_e, T_e)=\sum_{u \in S_e} \sum_{v \in T_e} c(u,v)$ is the capacity of the cut in $G$.    
\end{itemize}
The {\em capacity} of an edge $uv$ of $T$ is equal to $\lambda_{uv}.$ 
\end{definition}

\begin{theorem}\label{thm:GH}\cite{GomoryH}
Every weighted graph $(G,c)$ has a Gomory-Hu tree which can be constructed in polynomial time.
\end{theorem}

For an unweighted graph $G = (V, E),$ we can introduce weights by setting \bmp{$c(u,v)=1$} for every $uv\in E.$ 
Let $T$ be a Gomory-Hu tree of $G$.
Then, by Definition~\ref{defn:gomory-hu-tree}, for every pair of vertices $u, v \in V(G)$, the size of a minimum edge cut between $u$ and $v$ in $G$ is 
the minimum capacity of an edge cut between $u$ and $v$ in $T$.
For $i \in [p]$, consider the set $E_i$ of all the edges of total capacity at most $i-1$ in $T$.
Deleting $E_i$ disconnects $T$ into several subtrees.
We call the vertex set of each such subtree an {\em $i$-segment} in $G$.
Thus, a subset of vertices $S \subseteq V(G)$ is an $i$-segment in $G$ if and only if for every $u,v \in S$, there are at least $i$ edge-disjoint paths between $u$ and $v$ in $G$ and $S$ is a maximal such subset.
It is obvious from the construction that the $i$-segments of $G$ form a partition of the vertex set of $G$,
  and can be computed in polynomial time.

Now let $G = (V, E)$ be an undirected graph, and $X \subseteq V(G)$.
For $i \in [p]$, let an \emph{$i$-block of $X$ in $G$} be a maximal subset $X' \subseteq X$
such that for every $u, v \in X'$, there are at least $i$ edge-disjoint paths between $u$ and $v$ in $G$.
We can use the Gomory-Hu tree to compute the $i$-blocks of $X$ in $G$, as follows.

\begin{lemma}
\label{lemma:i-block-gomory-hu-tree-connection}
Let $G = (V, E)$ be an undirected graph, and $X \subseteq V(G)$.
The $i$-blocks of $X$ in $G$ are precisely the sets $X \cap S$ over all $i$-segments $S$ of $G$. 
\end{lemma}

\begin{proof}
  Let $T$ be the Gomory-Hu tree of $G$, and let $u, v \in X$ be distinct
  vertices. By the definition of a Gomory-Hu tree, $\lambda_{uv}(G) \geq i$ if and
  only if there is no edge $e$ on the path from $u$ to $v$ in $T$ with
  capacity $c(e)$ less than $i$.  Since this is also equivalent to $u$ and $v$ being
  in the same $i$-segment of $G$, the statement follows.
\end{proof}

%
%

Based on Lemma~\ref{lemma:i-block-gomory-hu-tree-connection}, we can compute the collection of $i$-blocks in polynomial time using a Gomory-Hu tree of $G$.


\begin{lemma}
\label{lemma:i-blocks-laminar-family}
Let $G$ be a graph, $X \subseteq V(G)$, and $p$ be a fixed integer.
Then the collection of all $i$-blocks of $X$ for all \bmp{$i \in [p]$} forms a laminar family.
\end{lemma}
\begin{proof}
  Note first that for every $i$, the $i$-segments of $G$ form a partition of
  the vertex set, hence similarly, the $i$-blocks of $X$ form a partition of $X$
  for every $i \in [p]$. Furthermore, let $i, j \in [p]$ with $1 \leq i < j \leq p$. 
  It is obvious from the definition that the $j$-segments of $G$ form a refinement
  of the $i$-segments of $G$, since they   are formed from the Gomory-Hu tree
  by deleting an additional set of edges. 
  Hence the $j$-blocks of $X$ also form a refinement of the $i$-blocks of $X$ in $G$
  by Lemma~\ref{lemma:i-block-gomory-hu-tree-connection}, and the statement follows.
\end{proof}

\bmp{The approximation algorithm exploits the Gomory-Hu tree, which can be used to derive the $p$-blocks which capture the edge-connectivity of the graph.}
The proof of Lemma \ref{lemma:i-blocks-laminar-family} leads to Algorithm \ref{algo-compute-laminar-tree} which sets all $i$-blocks of $X$, $i\in [p]$, as nodes of a (laminar) tree. 
Algorithm \ref{algo-compute-laminar-tree} works as follows.
For all $i \in [p]$, it computes all the $i$-blocks of $X$ in $G$.
\bmp{Then, it sets $X$ as the root of the tree, and makes all $1$-blocks children of $X$.
After that, for every $1$-block $X$ and $2$-block $Y$, it sets $Y$ as a child of $X$ if $Y \subseteq X$.
It repeats this process for $2$-blocks, $3$-blocks, up to and including $p$-blocks (in this order).}

\begin{algorithm}[ht]
\label{algo-compute-laminar-tree}
	\caption{{${\sf LaminarTree}(G = (V, E), X, p)$}}
	\SetKwInOut{Input}{input}\SetKwInOut{Output}{output}
	\Input{$G = (V, E), X \subseteq V(G), p$}
	\Output{A laminar tree of $X$ in $G$}
	\emph{Initialize a tree $T := \emptyset$}\;
	\For{$i = 1,\ldots,p$}
	{
		\emph{Compute the set of all $i$-blocks of $X$ in $G$}\;
		\emph{$\cA_i \leftarrow$ the set of all $i$-blocks of $X$ in $G$}\;
	}
	\emph{Set $X$ as the root of $T$}\;
	\emph{Make all $1$-blocks the children of $X$}\;
	\For{$i=2,\ldots,p$}
	{
		\For{every $Y \in \cA_i$}
		{
			\emph{Make $Y$ a child of $W \in \cA_{i-1}$ in $T$ such that $Y \subseteq W$}\;
		}
	}
	\emph{Output $T$ as the laminar tree of $X$ in $G$}\;
\end{algorithm}


\begin{lemma}
	\label{lemma:improving-cut-size}
	Let $G$ be a graph, $u\in V(G)$ and let $A, B \subseteq V(G - u)$ such that $A \cap B = \emptyset$.
	If the size of a minimum $(A, B)$-cut in $G - u$ is $i$ and
	$\min\{|N(u) \cap A|, |N(u) \cap B|\} = j$ then the size of a minimum $(A, B)$-cut in $G$ is at least $i + j$.
\end{lemma}

\begin{proof}
	Consider a minimum $(A, B)$-cut $(S, T)$ in $G$; $A \subseteq S$ and $B \subseteq T$.
	Assume without loss of generality that $u \in S$.
	Suppose that the size of $(S, T)$ is at most $i + j - 1$.
	Consider $(S\setminus \{u\},T),$ which is a cut in $G-u.$
	Since $|N(u) \cap B| \geq j$, we have that the size of the cut $(S\setminus \{u\},T)$ is at most $i - 1$,
	but this contradicts the fact that the size of a minimum  $(A, B)$-cut in $G - u$ is $i$.
\end{proof}

Our approximation algorithm, Algorithm~\ref{approx-algo-p-edge-cvc}, proceeds as follows.
\begin{itemize}
	\item Let $L$ \bmp{denote} the set of vertices with degree less than $p$. \bmp{By Lemma~\ref{lemma:pvcvc-existence}, $G$ has a $p$-edge-connected vertex cover if and only if~$V(G) \setminus L$ is such a vertex cover. Hence if it is not, we output}
    that `$G$ has no feasible solution'.
	\item Compute a maximal matching $M$ of $G - N_G[L]$ and initialize $X = N_G(L) \cup V(M)$ (the vertices matched by $M$) and initialize $Y = X$ as the partial solution.
	\item Invoke Algorithm~\ref{algo-compute-laminar-tree} to construct the laminar tree $T$ of $X$ in $G[X]$.
	\item As long as $T$ has at least two leaves, \bmp{the fact that~$G$ has a $p$-edge-connected vertex cover implies that} there is \bmp{a vertex} $u \in V(G) \setminus (Y \cup L)$ whose neighborhood intersects two distinct $p$-blocks in $T$. We add $u$ into the \bmp{solution~$Y$} and recompute the laminar tree $T$ of $X$ in $G[Y]$.
	\item Output $Y$ as the solution.
\end{itemize}

\begin{algorithm}[ht]
\label{approx-algo-p-edge-cvc}
	\SetKwInOut{Input}{input}\SetKwInOut{Output}{output}
	\caption{Approximation algorithm for {\pcvc}}
	\Input{$G = (V, E)$}
	\Output{An approximate $p$-edge-connected vertex cover of $G$}
	\emph{$L \leftarrow \{v \in V(G) \mid \deg_G(v) < p\}$}\;
	\If{\bmp{$V(G) \setminus L$} is not a $p$-edge-connected vertex cover of $G$}
	{
		\emph{Output ``No feasible solution''}\;
	}
	
	\emph{Compute a maximal matching $M$ of $G - N_G[L]$}\;
	\emph{$X \leftarrow N_G(L) \cup V(M)$}\;
	\emph{$T \leftarrow {\sf LaminarTree}(G[X], X, p)$}\;
	\emph{$Y \leftarrow X$}\;
	\While{$T$ has at least two leaves}
	{
            \emph{Let $u \in V(G) \setminus (Y \cup L)$ s.t. $N_G(u)$ intersects two distinct $p$-blocks in $T$}\;
			\emph{$Y \leftarrow Y \cup \{u\}$}\;
			\emph{$T \leftarrow {\sf LaminarTree}(G[Y], X, p)$}\;
			}
	\emph{Output $Y$ as a solution}\;
\end{algorithm}

Our `recompute laminar tree of $X$ in $G[Y]$' ensures that at the end, every pair of vertices in $X$ \bmp{can be connected by $p$ edge-disjoint paths}.
Finally, based on the other characteristics, we are ensured that $G[Y]$ actually becomes $p$-edge-connected (proof as part of Theorem~\ref{approx-algo-p-edge-cvc} proof).

We are ready to prove the main result of this section, Theorem~\ref{thm:approx-algo-pecvc}, \bmp{which we restate for completeness.}

\ApproxAlgoPEdgeCVC*

\begin{proof}
We will show that Algorithm~\ref{approx-algo-p-edge-cvc} is a $2(p+1)$-approximation algorithm. 

Observe that the first output ``No feasible solution'' is correct due to Lemma \ref{lemma:pvcvc-existence}. 
Note $X=N_G(L) \cup V(M)$ is a vertex cover of $G.$ Suppose that  $T$ has more than one leaf and there is no vertex  $u\in V(G) \setminus (Y \cup L)$ such that 
$N_G(u)$ intersects two distinct $p$-blocks in $T.$ Then even adding all vertices of  $V(G) \setminus (Y \cup L)$ to $Y$ will not make $Y$ \bmp{into a}  $p$-edge-connected vertex cover of $G$ since there will be the same 
number of $p$-blocks of $X$ in $G[Y]$ before and after the addition because the vertices of $V(G) \setminus (Y \cup L)$ are not vertices of $M$ and thus form an independent set. However, this is impossible
as \bmp{$V(G) \setminus L$} is a $p$-edge-connected vertex cover of $G$. Thus, as long as $T$ has more than one leaf there is a vertex  $u\in V(G) \setminus (Y \cup L)$ such that $N_G(u)$ intersects two distinct $p$-blocks in $T$.

When $T$ has just one leaf, $G$ has only one $p$-block of $X$ in $G[Y]$. This means that for every pair $x,y$ of vertices in $X$ there are $p$ edge-disjoint paths in $G[Y]$ between $x$ and $y$. Let $u,v\in Y\setminus X$. Since $|N_G(u)|\ge p$ and $|N_G(v)|\ge p$ by Menger's theorem, there are $p$ edge-disjoint paths in $G[Y]$ between $N_G(u)$ and $N_G(v)$ with distinct end-vertices and hence $p$ edge-disjoint paths in $G[Y]$ between $u$ and $v.$ Similarly, we can see that there are $p$ edge-disjoint paths in $G[Y]$ between $u$ and any $x\in X.$ Therefore, $Y$ is a $p$-edge-connected vertex cover of $G$.

Let us analyze how $T$ changes after $u\in V(G) \setminus (Y \cup L)$ is added to $Y.$ First we consider $T$  before $u$ is added to $Y$.
By the description of Algorithm~\ref{approx-algo-p-edge-cvc}, $u$ has neighbors in two distinct $p$-blocks $X_1, X_2$ of $X$ in $G[Y]$.
Let $X_0$ be the least common ancestor of $X_1$ and $X_2$ in $T$ and let $X_0$ be an $i$-block of $X$ in $G[Y]$. 
 Observe that $X_0$ has two children $X_a$ and $X_b$ such that $X_1\subseteq X_a$ and $X_2\subseteq X_b.$ Note that $X_a$ and $X_b$ are $(i+1)$-blocks of $X$ in $G[Y]$
 and the minimum size of a $(X_a,X_b)$-cut is $i$ (otherwise, $X_0$ is not  an $i$-block of $X$ in $G[Y]$). Now consider what happens just after $u$ is added to $Y$.
By Lemma~\ref{lemma:improving-cut-size}, the size of any $(X_a, X_b)$-cut increases by at least one. Thus, the minimum size of a $(X_a,X_b)$-cut becomes $i+1$ and so $X_a$ and $X_b$ become part of a 
new $(i+1)$-block of $X$ in $G[Y].$ Thus, the number of the nodes of $T$ decreases. 

Let us now bound the approximation factor of Algorithm~\ref{approx-algo-p-edge-cvc}. The number of leaves in $T$ becomes one only when there is just one node on each level of $T$, i.e., $T$ has $p+1$ vertices.
Initially, $T$ may have up to $p|X|+1$ nodes. Thus, at most $p(|X|-1)$ nodes will be added to $X$ before a solution $Y$ is obtained. Hence, $|Y|\le |X|+p(|X|-1)\le (p+1)|X|.$ Since $M$ is a maximal matching of $G - N_G(L)$, at least one endpoint of each edge of $M$ has to be in any vertex cover of $G$. Thus, $\OPT(G) \geq |N_G(L)| + |M|$, where $\OPT(G)$ is the minimum number of vertices in a $p$-edge-connected vertex cover of $G.$ Since $|X|= 2|M|+ |N_G(L)|,$ $|X|\le 2\OPT(G).$ Therefore, $|Y|\le (p+1)|X|\le 2(p+1)\OPT(G).$
\end{proof}

\section{Hardness proofs}\label{sec:nopolykernel}

We will first show that $p$-{\CVC} admits no polynomial kernel unless {\nka} using the fact that  {\CVC} admits no polynomial kernel unless {\nka} \cite{DomLS14}.

\begin{theorem}
\label{thm:pcvc-lowerbound}
\bmp{For every fixed~$p\geq 1$,} {\PVCVC} is NP-hard and does not admit a polynomial kernel \bmp{parameterized by the solution size~$k$} unless \rm{\nka}.
\end{theorem}
\begin{proof}
Let $(G,k)$ be an instance of {\CVC}.
We construct an instance of $p$-{\CVC} as follows.
We add $p-1$ new vertices $v_1,\ldots,v_{p-1}$ to $(G, k)$ and edges $\{uv_i | i \in [p-1], u \in V(G)\}$ obtaining a new graph $G'$.
Observe that $G'$ can be computed in polynomial time.
Furthermore, $(G,k)$ is a yes-instance of {\CVC} if and only if $(G',k+p-1)$ is a yes-instance of $p$-{\CVC}.

{This reduction runs in polynomial time and $p$ is a constant. Hence, it satisfies the conditions of Definition~\ref{defn:PPT}. Thus, this it is a polynomial parameter transformation from {\CVC} to {\PVCVC}. As {\CVC} does not admit a polynomial kernel unless {\nka}, due to Proposition~\ref{prop:basic-kernel-lower-bound}, neither does {\PVCVC} unless {\nka}. This completes the proof of the kernel lower bound. Since {\CVC} is NP-hard, this also proves NP-hardness of the problem.}
\end{proof}

The reduction in the next theorem is from {\sc Red Blue Dominating Set}. In the problem, given a bipartite graph $G$ with partite sets $R$ and $B$, we are to decide whether there is $R'\subseteq R$ such that $|R'| \leq k$ and $N_G(R')=B$.
Dom et al. \cite{DomLS14} proved that {\sc Red Blue Dominating Set} parameterized by $k + |B|$ does not admit a polynomial kernel unless {\nka}.

\begin{theorem}
\label{thm:pecvc-lowerbound}
\bmp{For every fixed~$p\geq 1$,} {\PECVC} is NP-hard and does not admit a polynomial kernel \bmp{parameterized by the solution size~$k$} unless \rm{\nka}.
\end{theorem}
\begin{proof}
\bmp{Since $p$-edge- and $p$-vertex-connectivity are equivalent for~$p=1$, the lower bound for~$p=1$ follows from Theorem~\ref{thm:pcvc-lowerbound}. In the remainder we consider an arbitrary~$p \geq 2$.}
Let $(G = (R \uplus B, E), k + |B|)$ be an instance of {\sc Red Blue Dominating Set} and let $B=\{b^1,b^2,\dots ,b^t\}.$
Without loss of generality, we may assume that for every $v \in R$, there exists $u \in B$ such that $uv \in E(G)$ (otherwise we can just delete $v$).
Similarly, we may assume that for every $v \in B$, there exists $u \in R$ such that $uv \in E(G)$ (otherwise, there is no feasible solution).
We will also assume that $t \ge p$ and $k \ge p$, {as the instance can otherwise be solved in polynomial time.} 
Construct  a new graph $H$ from $G$ as follows.
\begin{itemize}
	\item  Add a complete graph $K_p$ with vertex set $A=\{a_1,\ldots,a_{p}\}$ such that $V(G)\cap A=\emptyset$, and edges $a_i r$ for every $i\in [p]$ and $r\in R$. 
	\item Replace every vertex $b^j$ of $B$ by a complete graph $K_p$ with vertices $\{b^j_1,\dots ,b^j_p\}$ such that if $r b^j\in E(G)$ then $r b^j_i\in E(H)$ for every $i\in [p].$ 
	Thus, $B$ is replaced by $\hat{B}$ of size $pt.$
	\item Attach a pendant vertex to every vertex in $\hat{B} \cup A$.
	\item Set $k' = k + p(t+1)$.
	Note that since $p$ is a fixed constant, we have that $k'$ is $\cO(k + t)$. 
	\end{itemize}

To complete the proof, it suffices to prove that $(G, k+t)$ is a yes-instance of {\sc Red Blue Dominating Set} if and only if $(H, k')$ is a yes-instance of  {\PECVC}.

$(\Leftarrow)$ Let $S$ be a $p$-edge-connected vertex cover of $H$ such that $|S| \leq k'$.
Observe that $\hat{B} \cup A \subseteq S$ since all of the vertices in $\hat{B} \cup A$ have pendant neighbors.
Since every vertex in $\hat{B}$ has just $p-1$ neighbors in $\hat{B}$ and $G[S]$ is $p$-edge-connected, $S$ must contain at least one neighbor in $R$ of every 
vertex of $\hat{B}.$ Since $|\hat{B} \cup A| = p(t+1)$, $k'=|\hat{B} \cup A| +k$ and 
so there must be a subset $R'$ of $R$ of size at most $k$ such that 
$N_G(R')=B$. 
Hence, $(G, k)$ is  a yes-instance of {\sc Red Blue Dominating Set}.

$(\Rightarrow)$ 
Let $R^* \subseteq R$ with $|R^*| = k$ such that $N_G(R^*) = B$. Let $R^*=\{r_1,\dots ,r_k\}$, recall that $k\ge p.$
We claim that $A \cup R^* \cup \hat{B}$ is a $p$-edge-connected vertex cover of $H$. Clearly, it is a vertex cover, so
it remains to prove that $H[A \cup R^* \cup \hat{B}]$ is $p$-edge-connected. Let $u,v$ be vertices of $A \cup R^* \cup \hat{B}$.
It suffices to prove that there are $p$ edge-disjoint paths between $u$ and $v$ for every choice of $u$ and $v$. Subject to symmetry,
it suffices to consider six cases: 
\begin{align*}
u, v \in A;  u, v \in R^*; u, v \in \hat{B}; \\
u \in A, v \in R^*; u \in A, v \in \hat{B}; u \in R^*, v \in \hat{B}.
\end{align*}
Below we will consider these cases one by one.
\begin{description}
\item[$u, v \in A$.]  Without loss of generality, let $u=a_1$ and $v=a_2.$ Then $a_1a_2$, $a_1a_qa_2$, $3\le q\le p$ and $a_1r_1a_2$ are $p$ edge-disjoint paths between $u$ and $v.$
\item[$u, v \in R^*.$] Without loss of generality, let $u=r_1$ and $v=r_2$. 
Then $r_1a_ir_2$, $i\in [p]$ are $p$ edge-disjoint paths between $u$ and $v.$
\item[$u, v \in \hat{B}.$] We will first consider the subcase when $u,v$ are from the same clique in $\hat{B}$.  Without loss of generality, let $u=b^1_1$, $v=b^1_2$  and $r_1 b^1\in E(G)$. 
Then $b^1_1 b^1_2$, $b^1_1 b^1_q b^1_2, 3\le q\le p$ and $b^1_1 r_1 b^1_2$ are $p$ edge-disjoint paths between $u$ and $v.$ 

Now consider the subcase when $u,v$ are from  different cliques in $\hat{B}$.
Without loss of generality, let $u=b^1_1$, $v=b^2_1.$ 
For every $i\in [p],$ let $P_i=r_j$ if $b^1$ and $b^2$ are adjacent to a common vertex $r_j$ in $G$ and $P_i=r_{i'} a_i r_{i''}$ otherwise, where $r_{i'}b^1,r_{i''}b^2\in E(G).$ 
Then $b^1_1 P_1 b^2_1$, $b^1_1 b^1_i P_i b^2_i b^2_1$, $2\le i\le p$ are $p$ edge-disjoint paths between $u$ and $v.$

\item[$u\in A, v \in R^*$.] There are $p$ edge-disjoint paths between $u$ and $v$ since $A\cup \{v\}$ forms a clique with $p+1$ vertices.

\item[$u\in A, v \in \hat{B}$.] Without loss of generality, let $u=a_1$, $v=b^1_1$ and $r_1 b^1 \in E(G).$ Then $a_1 r_1 b^1_1$, $a_1a_ir_1b^1_ib^1_1$ ($2\le i\le p$) are $p$ edge-disjoint paths between $u$ and $v.$

\item[$u\in R^*, v \in \hat{B}$.] Without loss of generality, let $u=r_1$ and $v=b^1_1$. We will first consider the subcase when $r_1b^1\in E(G).$ Then $Q_1=r_1b^1_1$, $Q_i=r_1 b^1_i b^1_1$ ($2\le i\le p$) are $p$ edge-disjoint paths between $u$ and $v.$
Now consider the subcase when $r_j b^1\in E(G)$ for some $j>1.$ Then $r_1 a_i r_jQ'_i$ ($ i\in [p]$) are $p$ edge-disjoint paths between $u$ and $v,$ where $Q'_i=Q_i-r_1$ and $Q_i$ is defined in the previous subcase.
\end{description}

{This reduction runs in polynomial time. As $p$ is a constant, this reduction provides a polynomial parameter transformation. 
As {\sc Red Blue Dominating Set} parameterized by $k + |B|$ does not admit a polynomial kernel \bmp{unless {\nka}, by Proposition~\ref{prop:basic-kernel-lower-bound} {\PECVC} does not admit a polynomial kernel unless {\nka}}. \bmp{As {\sc Red Blue Dominating Set} is well-known to be NP-hard (cf.~\cite[\S 4.1]{DomLS14}), the NP-hardness result follows as well.}}
\end{proof}

\section{Conclusions}
\label{sec:conc}
 We \bmp{presented} time efficient polynomial sized approximate kernelization schemes (PSAKS) for both {\PECVC} and {\PVCVC}.
 We also \bmp{gave} a $\cO^*(2^{\cO(pk)})$ time algorithm for {\PECVC}.
 The approach we use in this FPT algorithm does not work for $p$-vertex-connectivity. 
 Hence, an interesting open problem would be to \bmp{determine whether there exists} a singly exponential FPT algorithm for {\PVCVC}.
 
 We also \bmp{obtained} a polynomial-time  $2(p+1)$-factor approximation algorithm for {\PECVC}.
 \bmp{Again, the main idea used in this algorithm does not work for {\PVCVC}, which appears to be the harder variant of the two.}
 \bmp{We mention without proof that it is possible to use the decomposition into biconnected components to derive a constant-factor approximation algorithm for {\BCVC} ($p=2$)} using an approach similar to {\PECVC}. However, it is currently unknown how to get such a result for {\PVCVC} for any arbitrary fixed $p \geq 3$.
 \bmp{Finally, another interesting open problem would be to obtain a size efficient PSAKS for {\PECVC} and for {\PVCVC}, but this is also open for {\CVC}.}


\end{document}